\documentclass[10pt, conference, letterpaper]{IEEEtran}
\makeatletter
\def\ps@headings{%
\def\@oddhead{\mbox{}\scriptsize\rightmark \hfil \thepage}%
\def\@evenhead{\scriptsize\thepage \hfil \leftmark\mbox{}}%
\def\@oddfoot{}%
\def\@evenfoot{}}
\makeatother
\usepackage{amssymb}
\usepackage{amsfonts}
\usepackage{graphicx}
\usepackage[cmex10]{amsmath}
\usepackage{times}
\usepackage{authblk}
\usepackage{algorithm}
\usepackage{algorithmic}
\usepackage{booktabs}
\usepackage{subfigure}
\usepackage{url}
\usepackage{color}
\usepackage[sort&compress,square,numbers]{natbib}

\usepackage{algorithm} 
\usepackage{algorithmic} 
\usepackage{amsmath}

\begin{document}

\title{Topology-Aware Node Selection for Data Regeneration in Heterogeneous Distributed Storage Systems}
\author[$\dag$]{Qingyuan Gong}
\author[$\dag$]{Jiaqi Wang}
\author[$\dag \natural$]{Yan Wang}
\author[$\dag$]{Dongsheng Wei}
\author[$\star$]{Jin Wang}
\author[$\dag$]{Xin Wang}
\affil[$\dag$]{School of Computer Science, Fudan University, Shanghai, China}
\affil[$\natural$]{School of Software, East China Jiaotong University, China}
\affil[$\star$]{Department of Computer Science and Technology, Soochow University, China}

\maketitle
\newtheorem{theorem}{Theorem}
\newtheorem{lemma}{Lemma}
\newtheorem{definition}{Definition}
\newtheorem{property}{Property}
\newtheorem{corollary}{Corollary}
\newtheorem{proposition}{Proposition}
\newtheorem{claim}{Claim}

\newcommand{\D}{\mathcal{D}}
\newcommand{\C}{\mathcal{C}}
\newcommand{\M}{\mathcal{M}}
\newcommand{\G}{\mathcal{G}}

\renewcommand{\S}{\mathsf{S}}
\newcommand{\DC}{\mathsf{DC}}
\newcommand{\Beta}{\pmb{\beta}}

\linespread{1.03}

\section{Abstract}

Distributed storage systems introduce redundancy to protect data from node failures. After a storage node fails, the lost data should be regenerated at a replacement storage node as soon as possible to maintain the same level of redundancy. Minimizing such a regeneration time is critical to the reliability of distributed storage systems. 
Existing work commits to reduce the regeneration time by either minimizing the regenerating traffic, or adjusting the regenerating traffic patterns, whereas nodes participating data regeneration are generally assumed to be given beforehand. However, such regeneration time also depends heavily on the selection of the participating nodes. Selecting different participating nodes actually involve different data links between the nodes. Real-world distributed storage systems usually exhibit heterogeneous link capacities. It is possible to further reduce the regeneration time via exploiting such link capacity differences and avoiding the
link bottlenecks.

\section{Introduction} \label{sec: introduction}
Distributed storage systems ({\em e.g.} GFS \cite{ghemawat2003google} and HDFS \cite{shvachko2010hadoop} )are widely used
in data centers. For such storage systems with thousands of nodes, node failures occur frequently. According to measurements on Facebook's warehouse cluster, the median number of machine-unavailability events is more than $50$ per day \cite{a2013solution}. Failed storage nodes usually corrupt their data blocks. Redundant data must be introduced for promoting both the data durability and the quality of service.

Replication is a widely-adopted method in industrial storage systems to enhance data reliability, where three or more replicas are
stored in the system \cite{ghemawat2003google, shvachko2010hadoop}. But this inevitably introduces a high storage overhead. To improve storage efficiency, erasure codes ({\em e.g.} \cite{weatherspoon2002erasure, rhea2003pond, bhagwan2004total, huang2012erasure}) have long been proposed to substitute the simple replication strategy in many designs.

A regenerating coded storage system is usually described as $(N,n,k,d,\alpha)$. $N$ is the total number of storage nodes in the system. The original file of size $M$ is divided equally into $k$ pieces, and these pieces of data are encoded into $n$ units to be stored separately on $n$ storage nodes in the system. Each storage node holds $\alpha$ coded data blocks. After a storage node fails, it is required to regenerate the same amount of data blocks at an alternative node (called {\em{newcomer}}) as soon as possible. The regeneration process usually includes downloading data from $d$ survival nodes (called {\em{providers}}), and encoding the redundant data at the newcomer.

In storage systems, when a node fails, redundant data will be lost and the original file will be unreliable. Considering the frequency of node failure, data repair should be arranged on a routine basis. Especially, in case more storage nodes beak down that less than $d$ survival nodes are available to reconstruct the original file, we need to regenerate the failed data as soon as possible to keep the redundancy level of the storage system.

For erasure codes based storage systems, Dimakis \emph{et al.} \cite{dimakis2010network} found that the amount of downloaded data can be reduced by pre-coding at each provider, which consequently reduces the regeneration time.
However, the regeneration time depends not only on the amount of downloaded data but also on the capacities of the download links. Link capacity in this paper means the available bandwidth between any two storage nodes for data regeneration. 

Data center networks are usually constructed in hierarchy and present heterogeneous characters \cite{ghemawat2003google, shvachko2010hadoop, fat-tree2009}. Available bandwidth between any two storage nodes varies greatly. In a Fat-tree based data center, for example, links between nodes in the same rack are able to provide much larger available bandwidths than inter-rack connections \cite{fatreee2008al}. Meanwhile, storage nodes may be assigned different services, generating different background traffics \cite{cloudcost2008cost, benson2010network, benson2010understanding}. The disparity of link bandwidths becomes even larger between storage servers of different generation \cite{googleArchitecture2003web}.

Existing works on minimizing the regeneration time assume that the set of providers and the newcomer are given in advance. However, we notice that this assumption imposes limitations on regeneration time reduction in practice, due to the fact that different participating nodes do affect the regeneration time. If the selection of providers or newcomer introduces low capacitated links into the repairing topology, the regeneration time may heavily be prolonged. This is prone to happen especially in heterogeneous networks when the providers and newcomer are selected randomly.


\begin{figure}[!htb]
\centering
\includegraphics[width = 9cm]{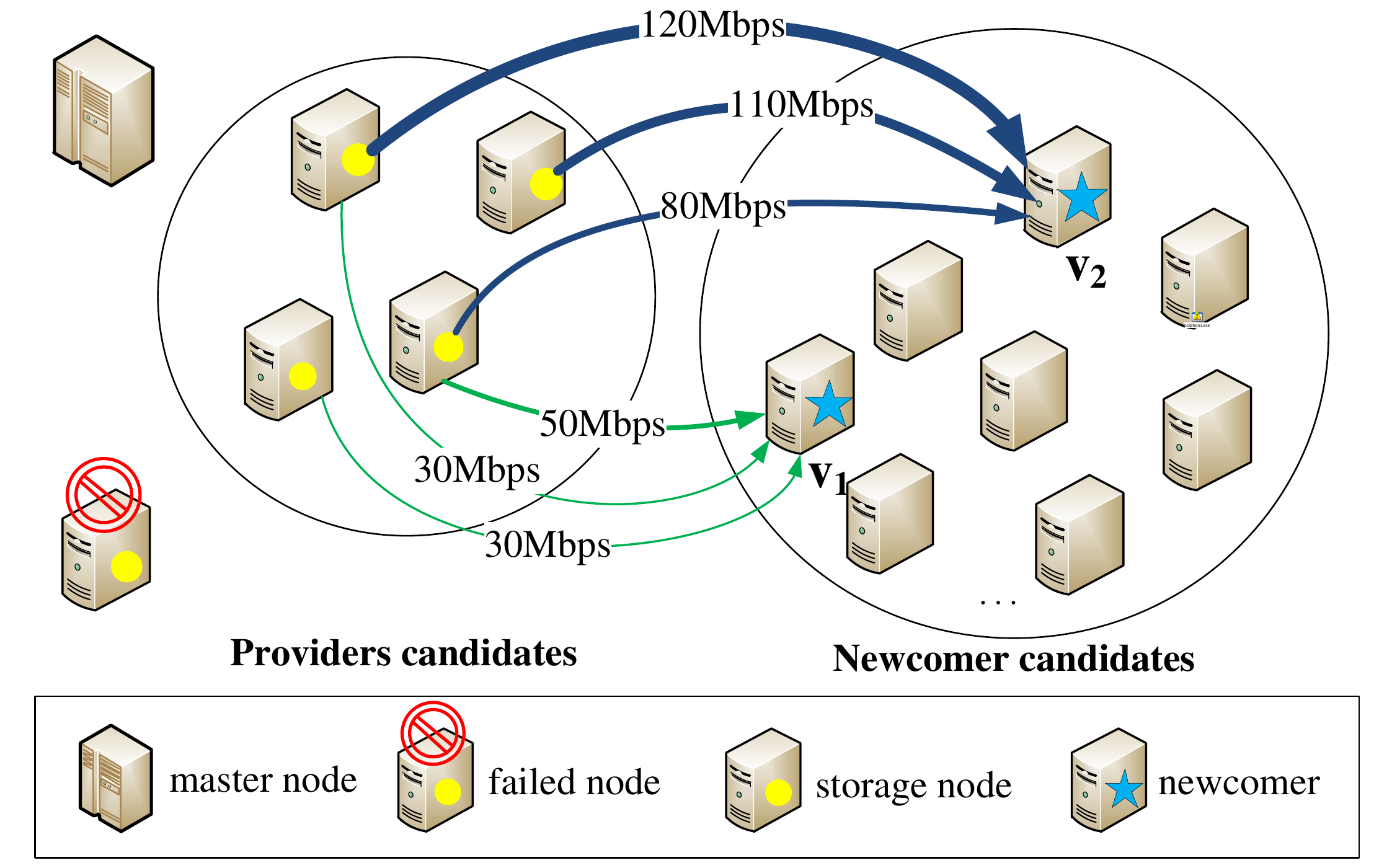}
\caption {Examples for various node selections. All survival nodes are providers candidates and others are newcomer candidates. $M= 480$ Mb. $\alpha = M/k = 240$ Mb. $\beta$=$\frac{\alpha}{d-k+1} = 120$ Mb. Available bandwidth of links in the topology is determined by the selection of the newcomer and providers. The regeneration time of the topologies with $v_1$ and $v_2$ being the newcomer is 4s and 1.5s, respectively.}
\label{Fig:6}
\end{figure}


The situation can be illustrated with the example shown in Fig.~\ref{Fig:6}. We consider an overlay network along with a master node, 
and assume that the link capacity between any two storage nodes obeys a uniform distribution on the interval $[10,120]$Mbps, which is commonly used in the literature \cite{lee2005measuring}. 
Suppose an original file of $M = 480Mb$ is stored based on a $(n=5,k=2)$ MDS (Maximum Distance Separable) code. MDS code requires $n = 5$ storage nodes to store the coded data, each holding $\alpha = M/k = 240Mb$; and accessing any 2 out of the 5 storage nodes is able to reconstruct the file. We set $d = 3$ for data regeneration that the newcomer should download data from $3$ providers to regenerate the failed data. If a storage node fails, any other nodes not involved in storing the original file could be the newcomer, forming the newcomer candidates set on the right of the figure. The remaining $n-1 = 4$ survival nodes constitute the set of provider candidates on the left. A newcomer needs to download $\beta = \frac{\alpha}{d-k+1} = 120Mb$ from each of the $d = 3$ providers to regenerate the lost data blocks. 
Fig.~\ref{Fig:6} demonstrates two possible selections of the newcomer and providers. Under random selection, we possibly have the newcomer $v_1$ and the three providers with arrows pointing to $v_1$. Clearly, this random selection leaves a bottleneck bandwidth of 30Mbps. Then, the regeneration time turns out to be $\frac{\beta}{30Mbps} = 4$ seconds. However, through a specific selection scheme for deliberately avoiding low capacitated links in the regeneration progress, we can obtain the newcomer $v_2$ and three providers with arrows pointing to $v_2$. The bottleneck bandwidth of the repairing topology is improved to 80Mbps, and the regeneration time is then reduced to $\frac{\beta}{80Mbps} = 1.5$ seconds. Obviously, a proper node selection scheme can substantially improve the minimization of the regeneration time by finding high-capacitated links. 

Identifying the high-capacitated links in overlay networks may not need much effort. However, we need to guarantee that these selected high-capacitated links from such a considerable candidate set really form a repairing topology. The $n-1$ survival nodes are candidates for $d$ providers, and all storage nodes in the storage system excluding the $n$ holding coded blocks of the original file are newcomer candidates. In Fat-tree architecture, it's even harder to obtain the preferred links because a physical link tends to be shared by many data flows. Moreover, time complexity of the node selection algorithm must be limited to minimize the extra time cost.

In distributed storage systems, it is often the case that a few storage nodes manage and control all nodes in the cluster. For instance, each GFS cluster consists of a master node to hold file system metadata, maintaining property information about all chunks in the cluster \cite{ghemawat2003google}. HDFS employs NameNode to trace attributes of DataNodes, recording their status information, modifying, or accessing events \cite{shvachko2010hadoop}. Node selection for data regeneration is possible to be realized under the control of a master node.

In this paper, we study the problem of how to select the participating nodes in heterogeneous networks, to form a repairing topology, aiming to minimize the regeneration time. To the best of our knowledge, this is the first investigation of the node selection problem for data regeneration. 
Our contributions consist of three parts:
\begin{enumerate}
\item We jointly consider the selection of the newcomer and providers to minimize the regeneration time for distributed storage systems employing regenerating codes. We conduct experiments based on end-to-end available bandwidth from real network. 

\item We incorporate the observation of flexible regeneration and obtain enriched node selection algorithms, where the amount of data downloaded from each provider is dynamically determined by link capacities.
 Simulation results show that node selection scheme with flexible regeneration traffic can further reduce the regeneration time.

\item We also consider the hierarchical topology and shared links in real-world data center networks. We analyze the data regeneration time and propose a heuristic node selection algorithm with better time complexity. The efficiency of this algorithm is pronounced in our trace driven simulations.
\end{enumerate}


Node selection schemes in overlay networks have been studied in our previous work \cite{qygong2015}. In this paper, we extend node selection schemes to the regenerating topologies in the real-world data centers. Transmission paths of real-world hierarchical architecture are much more complicate. We first analyze how the regeneration time should be computed in Fat-tree topology, which is totally different in overlay networks. Optimal regeneration time can be obtained through the selection of participating nodes based on traversal or greedy algorithms in overlay networks. However, as demonstrated in our analysis, the time complexity will be enormous if we apply traversal method in Fat-tree topology. To settle this problem, we propose a heuristic algorithm with time complexity $O(n\lg n)$ to select the providers and the newcomer. Although there may exist rare cases that another repairing topology achieves better regeneration time than the one obtained by our algorithm, exceptions can be tolerated for the great reduction of time cost.

The remainder of the paper is organized as follows. Sec.~\ref{sec: Related Work} covers the related work about the determination of participating nodes for data regeneration in distributed storage systems. In Sec.~\ref{sec: Network Model}, we formulate the node selection problem. Sec.~\ref{sec: Algorithm} proposes the node selection algorithms and discuss its overhead. In Sec.~\ref{sec: star-flexible}, we enhance our algorithm by introducing the idea of flexible end-to-end traffic. Node selection problem in the real-world data center architecture is discussed in Sec.~\ref{sec: fat-tree}. The experimental results are in Sec.~\ref{sec: evaluation}. This paper is concluded in Sec.~\ref{sec: Conclusion}.

\section{Related Work} \label{sec: Related Work}

Data center networks(DCN) usually present a hierarchical topology consisting of routers, switches, and storage nodes \cite{fatreee2008al}. Servers are deployed in racks. Communications between servers must go through higher layer switches of the network. For example, in a GFS cluster, servers in a same rack can interconnect via a $100Mbps$ Ethernet access switch; while servers in different racks must go through a core gigabit switch \cite{googleArchitecture2003web}. 
In addition, distributed storage systems usually support a variety of applications. Service can be assigned dynamically to any storage node and generate different amount of network traffic \cite{cloudcost2008cost, benson2010network, benson2010understanding}. Competition between applications may enlarge the difference of available bandwidth between two storage nodes for data regeneration. For large distributed storage systems, servers last within two or three years \cite{googleArchitecture2003web}. Aged machines are much slower than current-generation ones.

Lee \emph{et al.} measured available bandwidths between nodes on PlanetLab \cite{lee2005measuring}. From their measurement results, the minimum bandwidth can be as small as $0.3Mbps$; and over a portion of 99.22\% falls in the interval $[0.3,120]Mbps$. Furthermore, to fully use the core switches of high access rate, oversubscription is widely introduced in real DCNs \cite{fatreee2008al, cloudcost2008cost}. Higher oversubscription means lower bandwidth servers can employ between an access switch of a rack and a core switch. Typical values of oversubscription are designed within a range from 2.5:1 to 80:1 for a gigabit core switch, to achieve a minimal link capacity between $400Mbps$ and $12.5Mbps$ \cite{cisco2007oversubscription, googleArchitecture2003web}.

Since node failure occurs frequently in distributed storage systems, redundancy is introduced based on three-way replication or coding to assure data reliability. In replicated systems, if a replica node breaks down, the lost data will be cloned to a substitute node from the other two survival nodes. The substitute node is chosen under these criteria: 1) utilization of the disk space should be balanced; 2) the number of recent replica creation should be limited for one node; 3) replicas of a same file should be stored across racks \cite{ghemawat2003google, shvachko2010hadoop}. The latter two aim to complete the repair task as fast as possible, and avoid the repair traffic overwhelming the client traffic in the meanwhile.

Because of the high storage overhead of three replicas, it has been replaced by coding redundancy methods in many system designs.
Data repair in erasure coded distributed storage mainly involves data transmission from $d$ providers to a newcomer. Regeneration time cost lies in the transmission process. Dimakis \emph{et al.} \cite{dimakis2010network} proposed the model to reduce the amount of data blocks needed to transmit for data recovery. In this model, a tradeoff curve is given between storage cost and repair bandwidth, leaving no attention on how to determine the newcomer among many candidates. There are two special points on the curve, called MSR (minimum storage regenerating) codes where storage efficiency is the best but repair bandwidth is not optimal, and  MBR (minimum bandwidth regenerating) codes where minimum repair bandwidth is achieved at the sacrifice of storage efficiency.

Although MBR point has been paid much attention to minimize the regeneration time, minimum repair bandwidth will not assure minimum regeneration time in heterogenous distributed storage systems. Considering the heterogeneous link capacity, Li {\em et al.} \cite{li2010tree} proposed to reduce the regeneration time by adaptively selecting the high capacitated links. In the above models, the number of data blocks transmitted from the providers are the same, satisfying both exact and functional repair. The regeneration time is affected by the link capacity alone. To explore further reduction of regeneration time, Wang \emph{et al.} \cite{yan2014infocom} supposed that the data blocks downloaded from each provider can be flexibly determined based on the available bandwidths of the links. They proposed Flexible Tree-structured Regeneration (FTR) scheme, achieving the approximate minimum regeneration time.

Two-rack model is proposed by Gast\'{o}n in \cite{gaston2013realistic}. In this model, nodes are placed in two racks. Available bandwidths of intra-rack and inter-rack links are different. Providers in this model are chosen simply according to the location of the newcomer, with a view to employing higher bandwidth links. The newcomer is still given previously.

In the process of regeneration time minimization, nearly all designs assume that the set of providers and the newcomer are given in advance. The impact of different providers and newcomer in heterogenous link capacity environment is neglected.

Exact repair and functional repair are two versions of repair considered in the literature \cite{dimakis2011survey}. Exact repair requires the lost encoded blocks to be restored exactly. Functional repair reduces this requirement, allowing the newly generated blocks to be different from the lost ones, as long as the MDS property is still maintained, {\em i.e.,} the original file can be reconstructed by accessing any $k$ out of $n$ storage nodes. The model proposed by Dimakis is built on functional repair. Regenerating codes on the curve can be achieved by random linear network codes. The two interesting points MSR and MBR have also been realized with the extra constraint of exact repair \cite{exactrepairMBR2009exact, exactMSR2010exact, exactMSR22010distributed, msrcode2015MSR}. Exact repair codes present much better properties than random-network-coding-based functional repair codes.


The data regeneration problem in distributed storage systems has been widely explored. But it is worth mentioning that, the issue of how to select the providers and the newcomer is generally ignored. We propose optimal node selection schemes, aiming to construct repairing topologies to minimize the regeneration time. We use network simulator to evaluate our algorithms. Simulation results are obtained more independently in this way.

\section{problem formulation}\label{sec: Network Model}
In this section, we first present the data regeneration process and analyze previous node determination schemes. Then we formulate the problem of node selection for data regeneration. 

Given a regenerating coded storage system $(N,n,k,d,\alpha)$, 
when a storage node fails, the $d$ out of $n-1$ storage nodes can be accessed to recover the lost data blocks at a newcomer. The master node is treated as a control server, other than a storage node. We call the failed storage node as $v_f$, and the selected newcomer as $v_n$. For regenerating codes with uniform end-to-end traffic, each provider needs to transmit $\beta$ blocks to the newcomer \cite{dimakis2010network}. 
To maintain the same level of redundancy, we assume the newcomer $v_n$ 
can only be selected from the other $N-n$ nodes in the system. 
The repairing topology consists of $d$ providers, a newcomer and links from the providers to the newcomer. 
 In some regenerating models, the failed device is mended or replaced with a new one and the newcomer $v_n$ is assumed to be identical to $v_f$ by default. As the failure of storage nodes occurs randomly, we regard this case equivalent to that the newcomer is randomly selected.

In this paper, we consider the situation of single node failure, and utilize the conventional star-structured repairing topology where each provider transmits coded data blocks directly to the newcomer. 
Our aim is to find a newcomer and $d$ providers to minimize the regeneration time. 
We use a complete graph $G(V, E)$ \cite{li2010tree, yan2014infocom} to represent the overlay network, where $V$ is the set of storage nodes in the distributed storage system and $E$ is the set of communication links between storage nodes. $G(V', E')$ is a subnetwork of $G(V, E)$, consisted of the selected $d$ providers and the newcomer. Let $V_n$ represent the set of newcomer candidates, $V_p$ represent the set of providers candidates, and $P$ represents the $d$ selected providers. $G(V', E')$ can also be expressed as a tuple $[P(v_n), v_n]$ to indicate the providers and newcomer of a repairing topology. To guarantee that a newcomer and $d$ providers can be found, we assume that $|V_n| > 0, |V_p| \geq d$. To provide the same level of redundancy, the $n-1$ storage nodes will not be considered to be a newcomer in the repairing process. Thus $V_n \cap V_p = \varnothing$ and $V_n \cup V_p = V - \{v_f\}$. Since there are $N$ storage nodes totally in the cluster, $|V_n|=N-n$ and $|V_p|=n-1$.

For any two nodes $u, v$ in $G(V', E')$, let ${c(u, v)}$ denote the link capacity of $(u, v)$ and $f(u, v)$ denote the number of data blocks transmitted on $(u, v)$. Link capacity means the available bandwidth between any two storage nodes for data regeneration. We can ignore the encoding time on providers and the decoding time on the newcomer, as these two operations can be pipelined with the data transmission \cite{li2010tree, yan2014infocom}. Given that data blocks are transmitted simultaneously from the $d$ providers to the newcomer, we may represent the regeneration time as
\begin{equation}\label{gqy_e1}
t =  max\{\frac{f(u, v)}{c(u, v)} \ |\ (u,v)\in E'\}.
\end{equation}

According to Equation (\ref{gqy_e1}), if each provider transmits $\beta$ blocks to the newcomer, the regeneration time is determined by the minimum link capacity between the $d$ providers and the newcomer, {\em {i.e.,}} the bottleneck bandwidth of the repairing topology. In distributed storage systems, link capacities between any two storage nodes vary over a wide range because of the hierarchical structure, oversubscription, and real-time service. We assume link capacities in the system obey a uniform distribution on an interval such as $[10, 120]$Mbps \cite{lee2005measuring}. Links in $G(V', E')$ are determined by the $d$ providers and the newcomer. 


In real distributed storage systems, storage nodes report their status information to the master node at a certain interval  \cite{ghemawat2003google, shvachko2010hadoop}. Similarly, we employ a master node to manage the node selection process, as shown in Fig.~\ref{Fig:6}. The end-to-end network bandwidth measurement is a challenging task, yet it can be settled through probing based techniques at a low time cost \cite{meas1jsac2003evaluation}. The master node obtains the available bandwidths between any two nodes in the overlay network $G(V, E)$ periodically, as other status information. Then the realtime end-to-end bandwidth information of all node pairs will be ready before node selection is needed. Details in available bandwidth measurement are treated as a blackbox, free to vary and not a primary concern of this work. 
To facilitate further discussions, we summarize important notations in the paper for ease of reference in Table \ref{tab:1}.

\begin{table}[t]
\centering
\caption{Table of Nomenclature} \label{tab:1}
\begin{tabular}{c p{6cm}}
\hline
Notations & Definition\\
\hline
$G(V,E)$ & The distributed storage system topology\\
$G(V',E')$ & The repairing topology\\
$M$ &   Size of the original file\\
$N$ & Total number of storage nodes in the system\\
$n$ & Number of storage nodes holding coded file blocks\\
$d$ & Number of providers  \\
$\alpha$ & Number of coded blocks on each of the $n$ storage nodes \\
$\beta$ & Number of coded blocks generated by each provider \\
$v_f$ & The failed storage node\\
$V_p$ & Providers candidates, i.e., survival nodes \\
$V_n$ &  Newcomer candidates, i.e., nodes in $V$ excluding $\{v_f\} \cup V_p$\\
$\mathcal{B}$ & Sorted bandwidth sequence of links between $V_p$ and $V_n$\\
$v_n$ &  The selected newcomer \\
$\mathcal P(v_n)$ & The selected $d$ providers when the newcomer is $v_n$ \\
$[P(v_n), v_n]$ & The repairing topology with $v_n$ and $P(v_n)$ to be the newcomer and providers respectively \\
$\mathbf{B}(\mathcal{P}_i(j),j)$ & Available bandwidth of links between provider $\mathcal{P}_i(j)$ and newcomer $j$\\
\hline
\end{tabular}
\end{table}


\section{problem formulation}\label{sec: Network Model}
In this section, we first describe the data regeneration process and analyze previous node determination schemes. Then we formulate the problem of node selection for data regeneration. 

Given a regenerating coded storage system $(N,n,k,d,\alpha)$, 
when a storage node fails, the $d$ out of $n-1$ storage nodes can be accessed to recover the lost data blocks at a newcomer. The master node is treated as a control server, other than a storage node. We call the failed storage node as $v_f$, and the selected newcomer as $v_n$. For regenerating codes with uniform end-to-end traffic, each provider needs to transmit $\beta$ blocks to the newcomer \cite{dimakis2010network}. 
To maintain the same level of redundancy, we assume the newcomer $v_n$ 
can only be selected from the other $N-n$ nodes in the system. 
The repairing topology consists of $d$ providers, a newcomer and links from the providers to the newcomer. 
 In some regenerating models, the failed device is mended or replaced with a new one and the newcomer $v_n$ is assumed to be identical to $v_f$ by default. As the failure of storage nodes occurs randomly, we regard this case equivalent to that the newcomer is randomly selected.

In this paper, we consider the situation of single node failure, and utilize the conventional star-structured repairing topology where each provider transmits coded data blocks directly to the newcomer. 
Our aim is to find a newcomer and $d$ providers to minimize the regeneration time. 
We use a complete graph $G(V, E)$ \cite{li2010tree, yan2014infocom} to represent the overlay network, where $V$ is the set of storage nodes in the distributed storage system and $E$ is the set of communication links between storage nodes. $G(V', E')$ is a subnetwork of $G(V, E)$, constituted by the selected $d$ providers and the newcomer. Let $V_n$ represent the set of newcomer candidates, $V_p$ represent the set of providers candidates, and $P$ represents the $d$ selected providers. $G(V', E')$ can also be expressed as a tuple $[P(v_n), v_n]$ to indicate the providers and newcomer of a repairing topology. To guarantee that a newcomer and $d$ providers can be found, we assume that $|V_n| > 0, |V_p| \geq d$. To provide the same level of redundancy, the $n-1$ storage nodes will not be considered to be a newcomer in the repairing process. Thus $V_n \cap V_p = \varnothing$ and $V_n \cup V_p = V - \{v_f\}$. Since there are totally $N$ storage nodes in the cluster, $|V_n|=N-n$ and $|V_p|=n-1$.

For any two nodes $u, v$ in $G(V', E')$, let ${c(u, v)}$ denote the link capacity of $(u, v)$ and $f(u, v)$ denote the number of data blocks transmitted on $(u, v)$. Link capacity means the available bandwidth between any two storage nodes for data regeneration. Without loss of generality, the encoding time on providers and the decoding time on the newcomer can be ignored, since these two operations can be pipelined with the data transmission \cite{li2010tree, yan2014infocom}. Given that data blocks are transmitted simultaneously from the $d$ providers to the newcomer, we may compute the regeneration time as
\begin{equation}\label{gqy_e1}
t =  max\{\frac{f(u, v)}{c(u, v)} \ |\ (u,v)\in E'\}.
\end{equation}

According to Equation (\ref{gqy_e1}), if each provider transmits $\beta$ blocks to the newcomer, the regeneration time is determined by the minimum link capacity between the $d$ providers and the newcomer, {\em {i.e.,}} the bottleneck bandwidth of the repairing topology. In distributed storage systems, link capacities between any two storage nodes vary over a wide range because of the hierarchical structure, oversubscription, and real-time service. We assume link capacities in the system obey a uniform distribution on an interval such as $[10, 120]$Mbps \cite{lee2005measuring}. Links in $G(V', E')$ are determined by the $d$ providers and the newcomer. 


In real distributed storage systems, storage nodes report their status information to the master node at a certain interval  \cite{ghemawat2003google, shvachko2010hadoop}. Similarly, we employ a master node to manage the node selection process, as shown in Fig.~\ref{Fig:6}. The end-to-end network bandwidth measurement is a challenging task, yet it can be settled through probing based techniques at a low time cost \cite{meas1jsac2003evaluation}. The master node obtains the available bandwidths between any two nodes in the overlay network $G(V, E)$ periodically, as other status information. Then the realtime end-to-end bandwidth information of all node pairs will be ready before node selection is needed. Details in available bandwidth measurement are treated as a blackbox, free to vary and not a primary concern of this work. 
To facilitate further discussions, we summarize important notations in the paper for ease of reference in Table \ref{tab:1}.

\begin{table}[t]
\centering
\caption{Table of Nomenclature} \label{tab:1}
\begin{tabular}{c p{6cm}}
\hline
Notations & Definition\\
\hline
$G(V,E)$ & The distributed storage system topology\\
$G(V',E')$ & The repairing topology\\
$M$ &   Size of the original file\\
$N$ & Total number of storage nodes in the system\\
$n$ & Number of storage nodes holding coded file blocks\\
$d$ & Number of providers  \\
$\alpha$ & Number of coded blocks on each of the $n$ storage nodes \\
$\beta$ & Number of coded blocks generated by each provider \\
$v_f$ & The failed storage node\\
$V_p$ & Providers candidates, i.e., survival nodes \\
$V_n$ &  Newcomer candidates, i.e., nodes in $V$ excluding $\{v_f\} \cup V_p$\\
$\mathcal{B}$ & Sorted bandwidth sequence of links between $V_p$ and $V_n$\\
$v_n$ &  The selected newcomer \\
$\mathcal P(v_n)$ & The selected $d$ providers when the newcomer is $v_n$ \\
$[P(v_n), v_n]$ & The repairing topology with $v_n$ and $P(v_n)$ to be the newcomer and providers respectively \\
$\mathbf{B}(\mathcal{P}_i(j),j)$ & Available bandwidth of links between provider $\mathcal{P}_i(j)$ and newcomer $j$\\
\hline
\end{tabular}
\end{table}

\section{Node Selection Scheme with uniform End-to-End Traffic} \label{sec: Algorithm}

%
Uniform end-to-end traffic means that all providers transmit the same amount of data blocks ({\em i.e., }$\beta$) to the newcomer. Participating nodes are selected to complete the repairing topology $G(V',E')$.

In this section, we present the node selection scheme called SPSN (Select Providers and Select Newcomer). The common design is also considered, where both the $d$ providers and the newcomer are determined randomly, called RS (random selection). We compare the regeneration time of these two schemes, to show the difference between random node selection, and optimal selection. Available bandwidth of a link in overlay networks is the bandwidth that can be reserved for data regeneration. We also prove that by applying node selection algorithm, we can obtain a star-structured topology with optimal regeneration time.

\subsection{Select Providers and Select Newcomer}

Assuming both the $d$ providers and the newcomer are not given in advance, we can optimally select $d$ providers from $V_p$ and a newcomer from $V_n$. The amount of data blocks transmitted from each provider to the newcomer is the same. Regeneration time $t$ of the resulting star-structured topology is determined by the bottleneck link capacities. 
Storage nodes in $V_p$, $V_n$, and the $(n-1)(N-n)$ links between them construct a bipartite graph.



The SPSN algorithm is shown in Alg.~\ref{Alg2}. At the beginning, links between providers and newcomer candidates are sorted according to their available bandwidths on the master node, forming the queue $\mathcal{B}$. To select the providers and the newcomer, we consider the largest capacitated link in the queue at each loop. The construction of the repairing topology is reflected by the in-degree of nodes in $V_n$. For each link $e(i,j)$, we increase the in-degree of the $e$'s endpoint $j$ (in $V_n$) by $1$; and put the endpoint $i$ (in $V_p$) into the corresponding providers set of node $j$, {\em i.e.,} $\mathcal{P}(j)$. In the meanwhile, check whether the in-degree of node $j$ achieves $d$. If so, $j$ becomes the newcomer we select from $V_n$, and $\mathcal{P}(j)$ are the providers selected from $V_p$; together with the links between them form the repairing topology $G(V',E')$. Otherwise, we need to keep on visiting links in queue $\mathcal{B}$.

\begin{algorithm}[htb]
\caption{Select Providers and Select Newcomer (SPSN) Algorithm}
\begin{algorithmic}[1]\label{Alg2}
\STATE $d^-(j) \leftarrow 0$  //in-degree of node $j$ in $V_n$
\STATE $\mathcal{P}(j) \leftarrow \varnothing, \forall j \in V_n$ ~//corresponding providers when node $j$ is selected as a newcomer
\STATE $\mathcal{B}\leftarrow \textbf{Sort}(\{\mathbf{B}(i, j)|i\in V_p, j \in V_n\}, \textbf{desc})$
\FOR{each $ t \in [1, |\mathcal{B}|]$}
\STATE $e(i,j) \leftarrow \mathcal{B}(t)$ ~//the current largest capacitated link
\STATE $d^-(j) \leftarrow d^-(j) + 1$ ~//increase the in-degree of the link's endpoint in $V_n$ by $1$
\STATE $\mathcal{P}(j) \leftarrow \mathcal{P}(j) \cup \{i\}$
\IF{$d^-(j) \geq d$}
\STATE $v_n \leftarrow j$
\RETURN $\mathcal{P}(v_n), \ v_n$
\ENDIF
\ENDFOR
\end{algorithmic}
\end{algorithm}

In the algorithm, we need to sort the $(n-1)(N-n)$ links between $V_p$ and $V_n$ in descending order. For a sort program, efficient polynomial-time algorithm exists, with the average complexity {\em $O(n\lg n)$} by {\em {quick sort}}. 
The algorithm is greedy in nature, and will not end until the degree of a newcomer candidate equals to $d$. 
All $(n-1)(N-n)$ edges should be read sequentially in the worst case. Thus, the algorithm runs in polynomial time {\em $O(n\lg n)$}. Next, we prove that the SPSN algorithm gives an optimal selection of the providers and newcomer.

\begin{theorem}
\label{theorem:1}
The regeneration time of the repairing topology with providers and the newcomer selected by SPSN is minimized.
\end{theorem}
\begin{proof}
We prove this theorem by way of contradiction. Assume that the providers and newcomer selected through SPSN is $[\mathcal{P}(v_n), \ v_n]$, and there exists another star-structured repairing topology $[\mathcal{P}(v_n'), \ v_n']$, which has a minor regeneration time than $[\mathcal{P}(v_n), \ v_n]$.

Alg.~\ref{Alg2} will stop as soon as it finds $[\mathcal{P}(v_n), \ v_n]$. Here we relax this condition and allow Alg.~\ref{Alg2} to go through all the links between $V_p$ and $V_n$ (equally means the parameter $t$ going though values in $(1, |\mathcal{B}|)$), and return every eligible result. Under this condition, both $[\mathcal{P}(v_n), \ v_n]$ and $[\mathcal{P}(v_n'), \ v_n']$ would be two elements of the result set.
As the results are generated one by one, we can define that $[\mathcal{P}(v_n), \ v_n]$ is obtained in step $s_1$ of the loop in Alg.~\ref{Alg2}, and $[\mathcal{P}(v_n'), \ v_n']$ in step $s_2$. From our assumption, $[\mathcal{P}(v_n), \ v_n]$ is firstly generated. So we understand $s_1 < s_2$.

According to Alg.~\ref{Alg2}, the last link added to the star-structured repairing topology determines the bottleneck bandwidth, and then the regeneration time. So $\mathcal{B}(s_1)$ and $\mathcal{B}(s_2)$ are the bottleneck bandwidth links for the topology $[\mathcal{P}(v_n), \ v_n]$ and $[\mathcal{P}(v_n'), \ v_n']$, respectively. Since links are sorted in descending order, we have $\mathcal{B}(s_1) > \mathcal{B}(s_2)$, which is contradictory to the original assumption.
\end{proof}
\subsection{Discussion about the algorithms}

\subsubsection{Coding patterns}As we have mentioned, regenerating codes can be realized as functional or exact repair, while functional repair is mainly built on random linear network codes. Coding patterns are not restricted when we refer to RS, FPSN, or SPSN. They only concern about how to select the storage nodes to form a repairing topology, without any change of the coded data blocks. Our goal is to minimize the time consumed to complete data regeneration, through previous selection of the participating nodes. Node selection itself is coding-independent. It applies to both functional repair and exact repair codes.

\subsubsection{Potential overhead}Our node selection algorithms are based on the available bandwidths of links between two storage nodes. Extra overhead compared with random selection is mainly introduced by available bandwidth estimation and computation. In distributed storage systems, chunk servers send their state information to the master node periodically. Available bandwidth measurement has received a great deal of attention for decades. Hu {\em{et al.}} proposed two available bandwidth techniques consuming only 1-2 seconds \cite{meas1jsac2003evaluation}; and Man {\em et al.} limited the overhead of the measurement to 2-4 RTT \cite{meas2003new}. The time overhead of obtaining the available bandwidth is minor compared to the regeneration time. A typical cluster in GFS contains about a thousand storage nodes \cite{ghemawat2003google}. Considering the state of the art computing power of CPU, the delay incurred by sorting links by bandwidth of such cluster scale can be ignored \cite{computing}.

\section{Node Selection Schemes with flexible end-to-end traffic} \label{sec: star-flexible}
In previous sections, we assumed each provider transmits the same amount of coded data blocks and designed SPSN algorithm. Uniform end-to-end traffic is a common assumption of regenerating codes, no matter for functional repair or exact repair \cite{dimakis2010network, dimakis2011survey, li2010tree, gaston2013realistic}. Wang {\em et al.} \cite{yan2014infocom} dropped this premise and proposed FTR (flexible tree-structured regeneration), allowing each provider to flexibly generate different amount of data blocks according to their available bandwidths for a given topology. 

Flexible amount of data blocks transmitted from each provider, expressed as $\beta^*$, is another crucial method to reduce the regeneration time in bandwidth heterogeneous networks. 
For a repairing topology obtained with random node selection, if we allow providers to transmit flexible amount of data blocks to the newcomer, it turns to be the approach FRS (random selection with flexible $\beta^*$). To further reduce the regeneration time, we extend the node selection schemes with the consideration of flexible number of data blocks transmitted from each provider, called FLEX (select $d$ providers and a newcomer with flexible $\beta^*$). Recall that regeneration time is defined as $t=max\{\frac{f(u, v)}{{c(u, v)}}\}$ in Equation (\ref{gqy_e1}), where $(u,v)$ enumerates links in the overlay network $G(V',E')$, and $f(u, v)$ denotes the data blocks transmitted on the link. When the amount of data blocks generated by each provider is different, the value of $f(u, v)$ is various for each link. Let $\beta^*_i$ denote the amount of data blocks generated by provider $i$, and $c_i$ represent the capacity of the link between the newcomer and the provider $i$, $i\in \{1,\cdots,d\}$. We can immediately transform the regeneration time $t$ as
\begin{equation}\label{gqy_e2}
 t=max\{\frac{\beta^*_i}{c_i}\}.
\end{equation}

The flexible method can further reduce the regeneration time, because providers of high bandwidth generate more data blocks while providers of low bandwidth generate less. Therefore, in the repairing topology, links of high bandwidth can be fully exploited and negative effects from links of low bandwidth are mitigated. To the best of our knowledge, existing realization of exact repair codes cannot produce different amount of data blocks at different providers. Exact repair needs to determine the coding matrixes to obtain the $\beta$ blocks previously, unable to change dynamically according to instant available bandwidths. Flexible amount of data blocks can be encoded from the $\alpha$ blocks stored at each provider through random linear network code, realized by deliberately choosing the random linear coding matrix each time. As a result, FRS and FLEX can only be applied to data regeneration based on random linear network code.

\begin{lemma}
\label{theorem:2}
Flexible regenerating codes show that, to minimize the regeneration time for a given repairing topology, the amount of blocks transmitted from provider $i$ is \cite{equation2014survey}:
\begin{displaymath}
\beta^*_i=\left\{
\begin{array}{ll}
\displaystyle \frac{c_i M}{k\sum_{r=1}^{d-k+1}c_r} & \textrm{if $ 1 \leq i \leq d-k+1 $} \\
\\
\displaystyle \frac{c_{d-k+1} M}{k\sum_{r=1}^{d-k+1}c_r} & \textrm{if $ d-k+1 <i \leq d  $}
\end{array} \right.
\end{displaymath}
\end{lemma}

where providers are labeled based on the capacity order of their links to the newcomer, as $c_1 \leqslant c_2 \leqslant \cdots \leqslant c_d$.

Alg. \ref{Alg:3} works in a way of traversal algorithm. At the beginning, we also sort the links $(i,j)$ between $V_p$ and $V_n$ according to their available bandwidths $\mathbf{B}(i, j)$, to obtain link sequence $\mathcal{B}$. For every newcomer candidate $j$ in $V_n$, we construct a repairing topology with its $d$ largest capacitated adjacent links $(\mathcal{P}(j),j)$, and endpoints $\mathcal{P}(j)$ in $V_p$. We apply Lemma \ref{theorem:2} to determine the amount of blocks generated by each provider, and compute the regeneration time of the topology by Equation (\ref{gqy_e2}). Finally we get the topology $[\mathcal{P}(v_n), v_n]$ with the minimum regeneration time. $\mathcal{P}(v_n)$ and $v_n$ become the selected providers and newcomer.

For a topology $[\mathcal{P}(j), j]$, the amount of data blocks is computed sequentially at each provider. To compute the amount of data blocks $\beta^*_i$ by Lemma~\ref{theorem:2}, we need to reverse the link sequence $\mathcal{B}(\mathcal{P}(j),j)$ and obtain a corresponding providers sequence $\mathcal{P}_i(j)$, such that the link capacity satisfies $\mathbf{B}(\mathcal{P}_1(j),j) \leqslant \mathbf{B}(\mathcal{P}_2(j),j) \leqslant \cdots \leqslant \mathbf{B}(\mathcal{P}_d(j),j)$. Transmission time on each link can be calculated as $\frac{\beta^*_i}{\mathbf{B}(\mathcal{P}_i(j),j)}$. Maximization of the expression will be the regeneration time of the repairing topology with node $j$ being the newcomer, according to Equation (\ref{gqy_e2}).

As $\mathbf{B}(\mathcal{P}_i(j),j)$ is sorted in ascending order, the $[1, d-k+1]$ part of the piecewise function will produce larger transmission time. Thus we obtain the regeneration time through the equation $t' \leftarrow \frac{M}{k\sum_{r=1}^{d-k+1}\mathbf{B}(\mathcal{P}_r(j), j)}$, in line 7 of Alg.~\ref{Alg:3}. After considering all provider candidates, the topology of minimum regeneration time will surely present.  

In Alg.~\ref{Alg:3}, $(n-1)(N-n)$ links between $V_p$ and $V_n$ need to be sorted in descending order firstly. The repeated step in the algorithm is to construct a repairing topology for each newcomer candidate, and compute its regeneration time. The algorithm runs in polynomial time {\em $O(n\lg n)$}.

\begin{algorithm}[htb]
\caption{Select Providers and Newcomer with Flexible $\beta^*$ (FLEX) Algorithm}
\label{Alg:3}
\begin{algorithmic}[1]
\STATE $t \leftarrow +\infty, P  \leftarrow \varnothing$
\STATE $\mathcal{B}\leftarrow \textbf{Sort}(\{\mathbf{B}(i, j)|i\in V_p, j \in V_n\}, \textbf{desc})$
\FOR{each $j \in V_n$}
\STATE $\mathcal{P}(j) \leftarrow \arg \textbf{Top}(\{\mathcal{B}(i, j)|i\in V_p\}, \textbf{d})$
\STATE Rearrange$(\mathcal{P}_1(j), \dots, \mathcal{P}_i(j))$ such that $\mathbf{B}(\mathcal{P}_1(j),j) \leqslant \mathbf{B}(\mathcal{P}_2(j),j) \leqslant \cdots \leqslant \mathbf{B}(\mathcal{P}_d(j),j)$
\STATE \begin{displaymath}
\beta^*_i=\left\{
\begin{array}{ll}
\displaystyle \frac{\mathbf{B}(\mathcal{P}_i(j), j)M}{k\sum_{r=1}^{d-k+1}\mathbf{B}(\mathcal{P}_r(j), j)} & \textrm{if $ 1 \leq i \leq d-k+1 $} \\
\\
\displaystyle \frac{\mathbf{B}(\mathcal{P}_{d-k+1}(j), j)M}{k\sum_{r=1}^{d-k+1}\mathbf{B}(\mathcal{P}_r(j), j)} & \textrm{if $ d-k+1 < i \leq d  $}
\end{array} \right.
\end{displaymath}
\STATE $t' \leftarrow \frac{M}{k\sum_{r=1}^{d-k+1}\mathbf{B}(\mathcal{P}_r(j), j)}$
\IF{$t > t'$}
\STATE $P \leftarrow \mathcal{P}(j), v_n \leftarrow j, t \leftarrow t'$
\ENDIF
\ENDFOR
\RETURN $P, v_n, \beta^*,t$

\end{algorithmic}
\end{algorithm}

\begin{theorem}
The algorithm FLEX is the optimal algorithm for $\beta$-flexible data regeneration.
\end{theorem}
\begin{proof}
We give a proof by way of contradiction. We assume the result of Alg.\ref{Alg:3} is $(P, v_n, \beta^*,t)$ , and there is another result $(\widetilde{P}, \widetilde{v}_n, \widetilde{\beta}^*, \widetilde{t})$ where$(P, v_n)  \neq  (\widetilde{P}, \widetilde{v}_n)$ and $\widetilde{t} < t$.
We need to consider two cases: $v_n = \widetilde{v}_n$ and $v_n \neq \widetilde{v}_n$.

1) When $v_n = \widetilde{v}_n$, then $P \neq \widetilde{P}$. We can see that in Alg. \ref{Alg:3}, when node $v_n$ is the newcomer, the regeneration time: $t \leftarrow \frac{M}{k\sum_{r=1}^{d-k+1}\mathcal{B}(\mathcal{P}_r(v_n), v_n)}$. Providers in $\mathcal{P}({v_n)}$ are sorted in ascending order of link capacities. $\sum_{r=1}^{d-k+1}\mathcal{B}{(\mathcal{P}_r(v_n), v_n)}$ is the sum of the $d-k+1$ minimum link capacities. In Alg.~\ref{Alg:3}, $\mathcal{P}({v_n)}$ is the endpoints of the $d$ links of largest bandwidth connecting $v_n$. Therefore $t$ is the minimum regeneration time if $v_n = \widetilde{v}_n$, {\em i.e.}, $t < \widetilde{t}$.


2) The case $v_n \neq \widetilde{v}_n$. When $\widetilde{v}_n$ turns to be the newcomer in Alg. \ref{Alg:3}, according to the same analysis as in 1), we have:
$t' = \frac{M}{k\sum_{r=1}^{d-k+1}\mathcal{B}(\mathcal{P}_r(\widetilde{v}_n)), \widetilde{v}_n)}$. Here $t'$ is the minimum repair time when $\widetilde{v}_n$ is the newcomer, thus $t' < \widetilde{t}$. However, $v_n$ is the optimal newcomer over all newcomer candidates in $V_n$, while $\widetilde{v}_n$ is just a one-shot optimal newcomer. So $t \leq t'  < \widetilde{t}$.

To sum up, we derive contradictions in both cases 1) and 2) to the assumption $t > \widetilde{t}$.

\end{proof}

\section{Node selection scheme in real-world data center networks}\label{sec: fat-tree}

Data center networks are built and employed to provide a diverse set of services, such as Internet-facing applications and data intensive applications etc. In cloud data centers that host Web and storage services, distributed storage systems are widely deployed for better data management. Quite a few research conclusions on distributed storage system have been applied to data centers \cite{a2013solution, fan2009HDFSRAIDdiskreduce, sathiamoorthy2013xoring}. Today's data center networks are usually organized hierarchically, including a tree of routing and switching devices. The topologies of data center networks possess special features. This can not be overlooked at the design of distributed storage systems.

Fat-tree topology is proposed by Al-Fares \emph {et al.}, aiming to realize load balance through spreading traffic over more links \cite{fatreee2008al}. It has been widely adopted by large scale data center networks. For example, Cisco applies Fat-Tree topology for efficient communication \cite{cisco2007oversubscription, fatreee2008al, fat-tree2009}. However, even in data centers where all application services spread across all racks, data traffics at each level still present different characterization \cite{benson2010network, benson2010understanding}. Physical bandwidths of links at each level differ orders of magnitude intuitively, moving up the network hierarchy. In this section, we consider the representative real data center topologies and explore the effect of the tiered architecture to the regeneration time.

\subsection{Regeneration time in Fat-tree Architecture}

Fat-tree topology has three layers of switches, {\em i.e.,} core, aggregation, and edge layer. Its scale is usually measured by the parameter $K$, the number of folks of core switches. When each core switch has $K$ folks, there are $K*(K/2)$ aggregation switches, edge switches, and $K*(K/2)^2$ storage nodes. An example of the interconnect architecture of Fat-tree with $K=4$ is shown in Fig. \ref{Fig:5}. In a Fat-tree based storage system $G(V, E)$, $V$ refers to the storage nodes, excluding all kinds of switches. A {\em link} in the Fat-tree architecture means the physical link between any two connected devices. The physical link sequence connecting any two storage nodes is defined as a {\em path}. $E$ is the set of path in $G(V, E)$. For star-structured repairing topologies $G(V', E')$ in this hierarchical architecture, data flows from $d$ providers are transferred to the newcomer at the same time, through layers of switches according to inherent routing policies. We assume each provider transmits the same amount of data.

\begin{figure}[!htb]
\centering
\hspace*{-0.10in}
\includegraphics[width = 9.5cm]{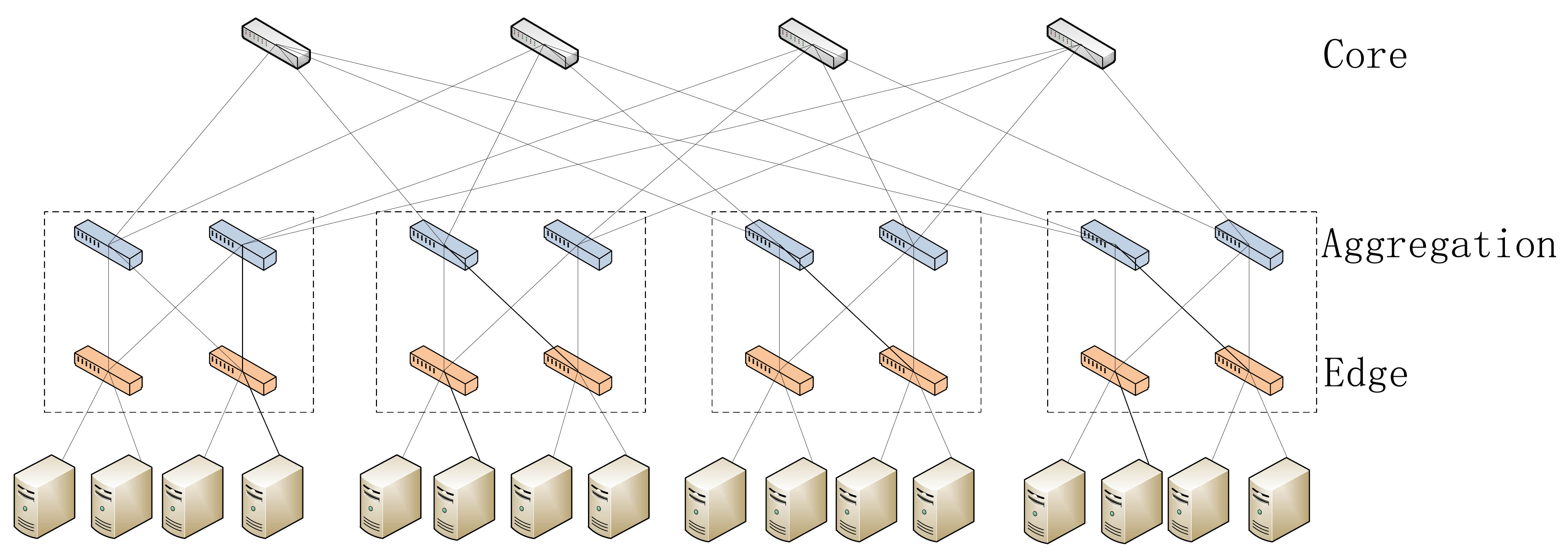}
\caption {The interconnect topology of Fat-tree with $K=4$. Storage servers are organized hierarchically with edge, aggregation, and core switches.}
\label{Fig:5}
\end{figure}

The regeneration time is determined by the available bandwidth of the paths in the repairing topology. However, links in Fat-tree architecture is prone to be shared by paths. Providers will not complete their transmission at the same time because of the heterogeneous available bandwidth of paths. Available bandwidth of the left providers will be affected if transmissions on other paths complete in advance. It's hard to derive the regeneration time of this complex topology.

For a repairing topology $G(V', E')$ in Fat-tree architecture, we use $l$ to denote physical links in each path, $c(l)$ to represent the available bandwidth of link $l$, and $q(l)$ signals the number of the path passing link $l$. Links of all paths in the repairing topology form a set $S$, while paths form the set $R$. When a physical link is occupied by more than one path, data flows will equally share the physical link bandwidth. The capacity of a path is decided by the minimal available bandwidth of links in the path. We compute the regeneration time $t$ of a given repairing topology through Alg.~\ref{Alg:4}.

\begin{algorithm}[htb]
\caption{Compute the regeneration time in Fat-tree architecture}
\begin{algorithmic}[1]\label{Alg:4}
\STATE $S \leftarrow$ all physical links in $G(V', E')$
\STATE $R \leftarrow$ $d$ paths in $G(V', E')$
\STATE $count[l]$  //the number of the path passing link $l$
\STATE $t \leftarrow 0$,~$t' \leftarrow 0$,~$\beta_i \leftarrow \beta$
\FOR{$j=1 ~to~ d$}
\STATE $R' \leftarrow R$
\STATE $q(l) \leftarrow count[l]$
\WHILE{$R' \neq \phi$}
\STATE $w$ $\leftarrow$ min $\{{\frac{c(l)}{q(l)}}\}$, $\forall~ l \in S$
\STATE $l=$ arg $min\{{\frac{c(l)}{q(l)}}\}$
\FOR {each $path_i \in R'$}
\IF {$l \in path_i$}
\STATE $w_i \leftarrow w$, $R' \leftarrow R'-\{path_i\}$
\STATE $c(l') \leftarrow c(l')-w_i$, $q(l') \leftarrow q(l')-1$~~$\forall~ l' \in path_i$
\ENDIF
\ENDFOR
\ENDWHILE
\STATE $t'\leftarrow min ~\{\frac{\beta_i}{w_i}\}$,  $i \leftarrow$ arg $min ~\{\frac{\beta_i}{w_i}\}$ ~~$\forall ~path_i \in R$
\STATE $count[l] \leftarrow count[l]-1$ ~~$\forall ~ l \in path_i$
\STATE deleting $path_i$ from $R$
\STATE $\beta_i \leftarrow \beta_i - t'*w_i$
\STATE $t\leftarrow t+t'$
\ENDFOR
\RETURN $t$
\end{algorithmic}
\end{algorithm}

In Alg. \ref{Alg:4}, the array $count[l]$ keeps the record the number of path passing the link $l$ and $R'$ to track the path set in the inner loop. Given $G(V', E')$, we first calculate the bandwidth share $w$ on each link, and obtain the link $l$ that offers the minimal bandwidth available to each path through it. If $l$ is a link of $path_i$, the available bandwidth $w_i$ of $path_i$ is then set to be $w$. For all other links $l'$ in $path_i$, we distract c(l') by $w$ and $q(l')$ by 1. Then we keep on checking the available bandwidth of all the $d$ paths based on the updated $R'$ and $q(l)$. As each provider transmits the same amount of data $\beta$, the path with the maximum available bandwidth will complete the transmission process firstly in the star-structured topology, which costs an interval $t'$. Then the path is dismissed from the path set $R$ and the capacity it occupied will be released. The left providers need to keep on transmitting their left amount of data. They may enjoy larger available bandwidth as some link capacities have been released. After all the $d$ providers complete data transmission, we obtain the regeneration time $t$.

\subsection{Select Providers and Select Newcomer in Fat-tree Architecture}

The description above presents the regeneration process intuitively. It can be seen that the regeneration time $t$ is affected by each participating node. Our aim is to select $d$ providers and a newcomer to obtain the optimal regeneration time. The optimal regeneration time can be attained through traversal. However, in Fat-tree architecture with the scale parameter $K$, the time cost to select the newcomer and $d$ providers can be calculated as $O(\dbinom{n}{d}K^3/4)$. Such time complexity is unacceptable for data regeneration. We propose the heuristic algorithm SPSN-F to selection the participating nodes. Correspondingly, the random selection of participating nodes in Fat-tree architecture is called RS-F.

In Fat-tree architecture, $d$ data flows from providers will reach the newcomer through the link between the newcomer and edge switch at the end of transmission paths. As the capacity share of links determines the available bandwidth of a path, it will be better to involve a newcomer with a larger capacitated link to the edge switch. For a repairing topology $(P, v_n)$, the regeneration time will be shorter if the bottleneck bandwidth of paths is larger, in coincidence with the overlay network, because both of them are many-to-one data transmission networks. Providers are selected by way of deletion. Details are shown in Alg. \ref{Alg5}.

\begin{algorithm}[htb]
\caption{Select Providers and Select Newcomer in Fat-tree Clusters (SPSN-F) Algorithm}
\begin{algorithmic}[1]\label{Alg5}
\STATE $R' \leftarrow$ paths between $V_p$ and $v_n$
\STATE $S' \leftarrow$ links of paths $\in R'$
\STATE $\eta_j=\overline{\eta_{j1}\eta_{j2}\ldots\eta_{ji}\ldots}$ // a binary number indicating whether a link $l_i$ is occupied by $path_j$
\STATE $\eta_{ji} \leftarrow 0$
\STATE $(u, v) \leftarrow$ the largest-capacitated link between edge switches and storage servers
\STATE $v_n \leftarrow~ v$
\STATE $P \leftarrow$ all provider candidates
\WHILE{$|P|>d$}
\STATE $w_i = \frac{c(l_i)}{q(l_i)}$~~$\forall l_i~\in~S'$
\STATE sort $(l_i)$ ~such that $w_1 \leq w_2 \leq \ldots$
\FOR{each $ path_j \in R'$}
\IF{$l_i \in path_j$}
\STATE $\eta_{ji} \leftarrow 1$
\ENDIF
\ENDFOR
\STATE $j \leftarrow$ the provider of the maximum value of $\eta_j$
\STATE deleting candidate $j$ from $P$
\STATE subtracting $q(l_i)$ by 1 for all links in $path_j$, deleting $l_i$ if $q(l_i)=0$
\STATE deleting $path_j$ from $R'$
\ENDWHILE
\RETURN $P, v_n$
\end{algorithmic}
\end{algorithm}

In Alg.~\ref{Alg5}, the newcomer is selected according to bandwidth between edge switches and all newcomer candidates. We consider the topology formed by all providers candidates, {\em i.e.,} $V_p$ and the newcomer $v_n$. A $path_j$ refers to the path from the provider candidate $j$ to $v_n$. To select $d$ providers, we calculate the available bandwidth can be distributed to a path on each link of the topology. A binary indicator variable $\eta$ is employed to measure the available link bandwidth of a path. Take the $i^{th}$ link in the ascending sequence $l_i$, if $l_i$ is in $path_j$, then the $i^{th}$ bit of $\eta_j$ is set to be $1$, {\em i.e.,} $\eta_{ji}=1$. Thus we obtain the value of $\eta_j$ for each provider candidate $j$. Delate the provider $j$ with the largest $\eta$ value, until there is $d$ providers in the set $P$, the algorithm ends.

The most time-consuming step of Alg.~\ref{Alg5} is the sorting operation in line 7-8. There are totally $K^3/4$ storage nodes in $G(V, E)$. And we need to check bandwidths of the links between $|V|- n$ storage nodes and their corresponding edge switches. All left steps are confined to $V_p$ and the selected newcomer, with a much smaller scale. Thus, the time complexity of this algorithm is $O(n\lg n)$.


%
%

\section{Evaluation} \label{sec: evaluation}

We perform a series of trace driven simulations to verify that our node selection schemes can indeed reduce the regeneration time. We answer the following questions. (1) How much time can be saved through optimal node selection? (2) How the efficiency of node selecion schemes changes in different heterogeneous environments? (3) What is the impact of the regenerating codes on node selection schemes? (4) Node selection schemes and flexible ene-to-end traffics, which is more efficient in regeneration time reduction?

Simulations are conducted on the network simulator 2 (NS2) to evaluate the effectiveness of different newcomer and providers selection schemes \cite{benson2010understanding}. The NS2 simulator is deployed on Ubuntu 11.04 and complied by gcc-4.3 and g++-4.3. 
We employ an overlay network and a Fat-tree based topology respectively in the simulation experiments; and build up TCP connections between storage nodes to transmit the repairing traffic. The queue management mechanism Drop-tail is used to manage the queue at routers \cite{patil2011droptail}. Under this mechanism, each packet is treated identically and newly coming packets will be dropped if the buffer is full. The queue limit is set to 50. Propagation delay on all links is set to 0.1ms. Providers generate data blocks using FTP. For traditional regenerating codes requiring providers to transmit equal amount of data blocks to the newcomer, each provider generates $\beta=\frac{M}{d(d-k+1)}$ blocks \cite{dimakis2010network}. For schemes requiring flexible amount of data blocks, each provider generates $\beta^*$ data blocks according to equation (\ref{gqy_e2}). We use an $(n=14,k=8)$ MDS code to introduce data redundancy and set $d=10$. According to real world network measurement, link capacity can be set to obey a uniform distribution on the interval $U(10,120)$Mbps \cite{lee2005measuring}. The size of the original file is 100Mb. Each simulation result is obtained after running $100$ times repeatedly.

The regeneration time is our main concern. 
In special, the encoding time on the providers and the decoding time on the newcomer are ignored because these two operations can be done simultaneously with data transmission \cite{yan2014infocom, li2010tree}. For evaluations on NS2, the regeneration time is measured as the time span from the start and the end of the data regeneration process, {\em i.e.,} the time interval from the first data block transmitted by a provider, to the last one received by the newcomer. 

Experiments are conducted in four phases. In the first part we examine the performance of the two node selection schemes RS and SPSN, showing the benefit of node selection. Different link capacity distributions are considered to test the impact of bandwidth heterogeneity on the schemes. In the second part, we test the effect of parameter settings such as the scale of the cluster and the coding patterns. The third part compares the regeneration time of FRS and FLEX, analyzing the efficiency of optimal node selection and flexible regeneration traffic approaches to reduce regeneration time. Finally, performances of FLEX and SPSN-F are shown.

\subsection{Benefit of node selection}

We simulate 7 different link capacity distributions: $U_1[0.3,120]$Mbps, $U_2[1,120]$Mbps, $U_3[10,120]$Mbps, $U_4[30,120]$Mbps, $U_5[50,120]$Mbps, $U_6[70,120]$Mbps, and $U_7[90,120]$Mbps. Experiment results show that the resulting topology with optimal node selection achieves much shorter regeneration time for the same original file. Fig.~\ref{Fig:1} shows the regeneration time of the repairing topology produced by RS and SPSN for a range of various available bandwidths.
\begin{figure}[!htb]
\centering
\hspace*{-0.15in}\includegraphics[width = 8cm]{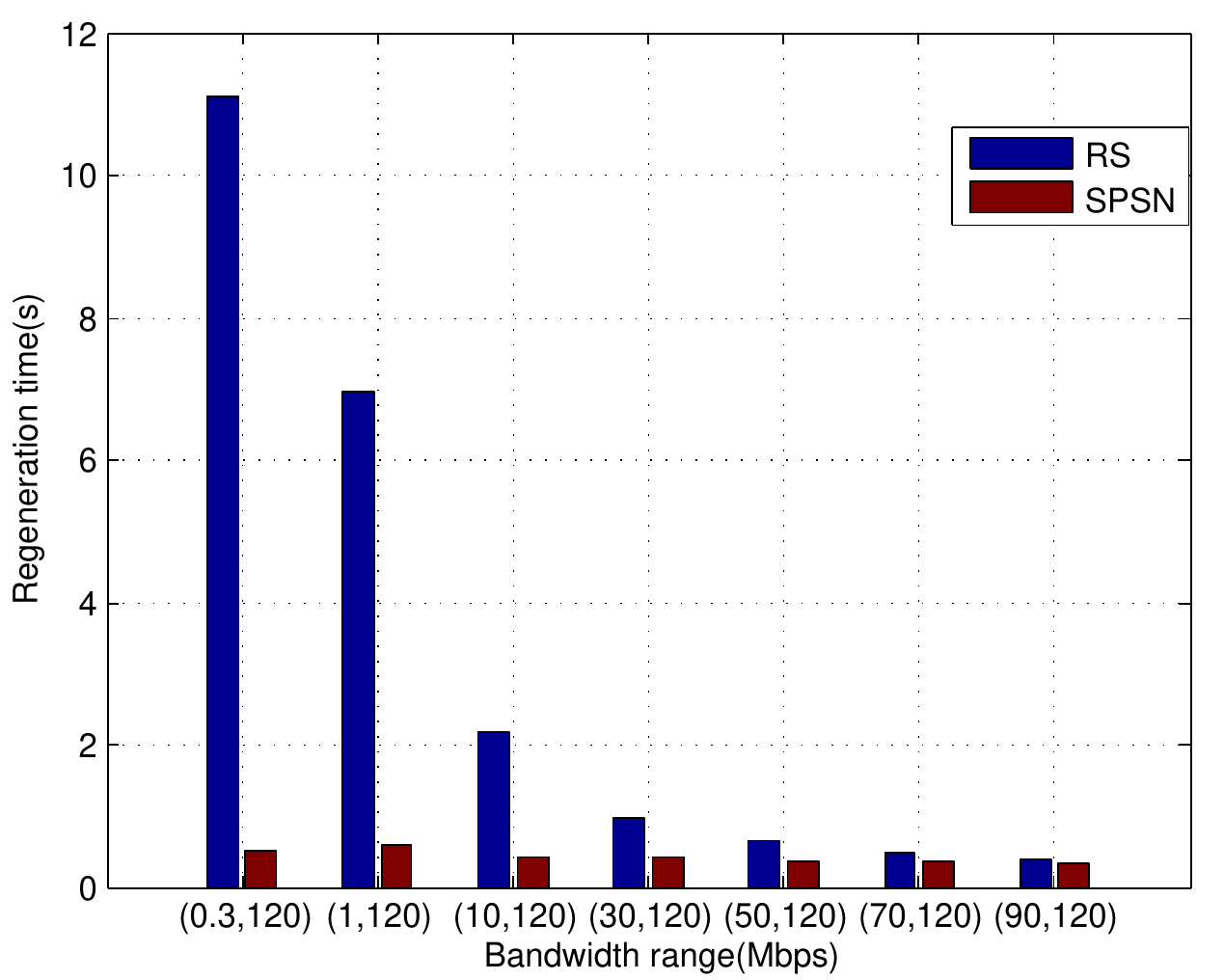}
\caption{Regeneration time of repairing topologies resulted from RS and SPSN respectively for different link capacity distributions. The parameters are N=1000, n=14, k=8, d=10. M=100Mb.}
\label{Fig:1}
\end{figure}

For the two schemes, providers generate the same amount of data blocks. Regeneration time is actually determined by the bottleneck bandwidth of the repairing topology. 
Given the overlay network $G(V,E)$, the probability that links of different capacities selected into $G(V',E')$ is almost the same for the scheme RS. Transmission bottleneck is likely to be lower capacitated links especially in bandwidth heterogeneous networks. When the newcomer and providers are allowed to be selected, we can avoid low capacitated links into the repairing topology deliberately. Thus SPSN can obtain relatively higher bottleneck bandwidth compared with RS. It can be seen from Fig.~\ref{Fig:1} that regeneration time reduces dramatically after node selection is introduced at each point of link capacity variance. The reduction becomes larger as the bandwidth variance turns larger.

RS performs much better as the link capacity variance turns weaker. On the other hand, 
for SPSN, we jointly select the $10$ providers from the 13 survival nodes and the newcomer $v_n$ from the candidate set $V_n$. We ought to obtain a repairing topology with more high capacitated links than RS. Fig.\ref{Fig:1} shows that SPSN performs even better than RS at each bandwidth variance. At the point that link capacities vary between $[10,120]$Mbps, SPSN finally generates a repairing topology with regeneration time reduction of nearly $80.74\%$ compared to RS. 

\subsection{Impact of system parameters}

We set the system scale as $1000$ storage nodes, and use an $(n=14,k=8)$ MDS code with $d=10$ providers in previous simulation. From Fig.\ref{Fig:1} we understand that the optimal selection of the newcomer and providers can obviously decrease the regeneration time, especially in largely heterogeneous networks. But $(14,8)$ MDS code is not universal and any $k \leq d \leq n$ works actually \cite{dimakis2010network}. For example, there exists a coding redundant storage system called Microsoft Azure Storage (MAS) employing $6$ storage nodes to recover the original file \cite{calder2011windows, huang2012erasure}. In this part of simulation, we will show the impact of the system scale and coding patterns on the regeneration time. Link capacity range in this part obeys the uniform distribution on the interval $[10,120]$Mbps.

\subsubsection{Impact of the number of storage nodes in the system}

We use $N$ to represent the number of nodes in the distributed storage system. More storage nodes in the system means more newcomer candidates as the value of $n$ and $d$ are fixed. However, node selection schemes always select high capacitated links as much as possible. The regeneration time will be affected more by link capacity range than the number of storage nodes in the system.

Results are shown in Fig.~\ref{Fig:2}. As the link capacity is uniformly distributed on an interval, for random selection, the size of the overlay network will not affect the probability of high capacitated links to be selected. We can see from Fig.~\ref{Fig:2} that the regeneration time of RS is fluctuating, but relatively stable at $2.06s$ for the bandwidth variance $[10,120]$Mbps. A similar trend presents: SPSN saves much regeneration time compared with RS. But as the network scale turns larger, regeneration time produced by SPSN decreases. This is because we are likely to obtain a better repairing topology with comparatively more newcomer candidates, as the total number of nodes increases in the system.


\subsubsection{Impact of the coding patterns of the system}

We use $(n=14,k=8)$ MDS code to introduce redundancy previously. In this part of experiments, we test the impact of coding patterns, {\em i.e.}, the value of $d$ and $n$, on the performance of our algorithms.

To test the impact of the number of providers, we set the value of $n$ normally as 20 and change the value of $d$ between $(8,19)$. Fig.~\ref{Fig:3} shows the impact of $d$ on regeneration time for the three schemes. From the figure we can see that the regeneration time of repairing topologies produced by both the schemes reduces as $d$ turns larger. Even the random scheme RS presents a regular decline, which coincides with the nature of regeneration codes. 
\begin{figure*}[!htbp]
   \subfigure[Different number of storage nodes in the system: coding parameters are n=14, k=8, d=10. ]{\includegraphics[width=0.333\textwidth]{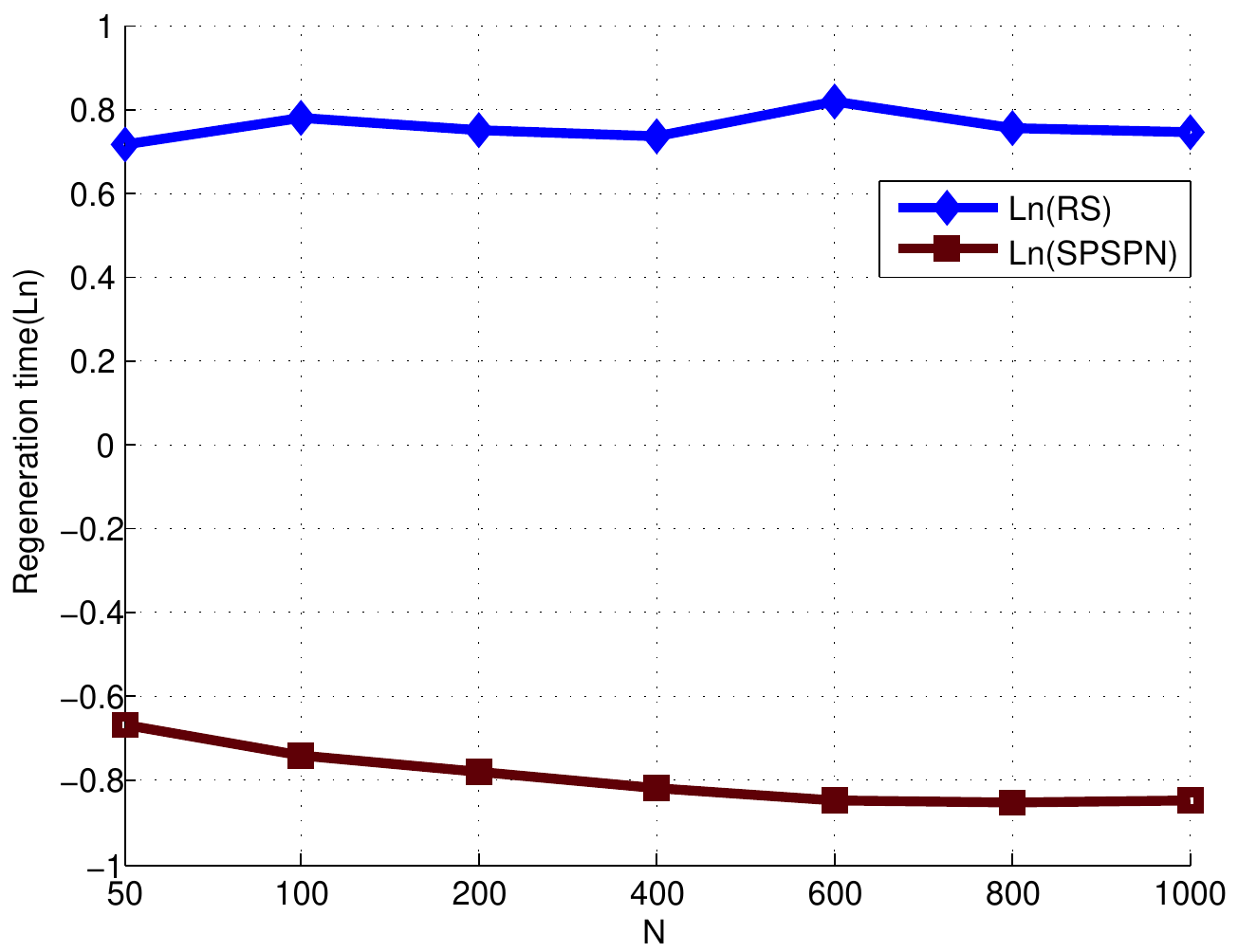}
   \label{Fig:2} }
  \subfigure[Different number of providers: other parameters are N=1000, n=20, k=8, ]{\includegraphics[width=0.333\textwidth]{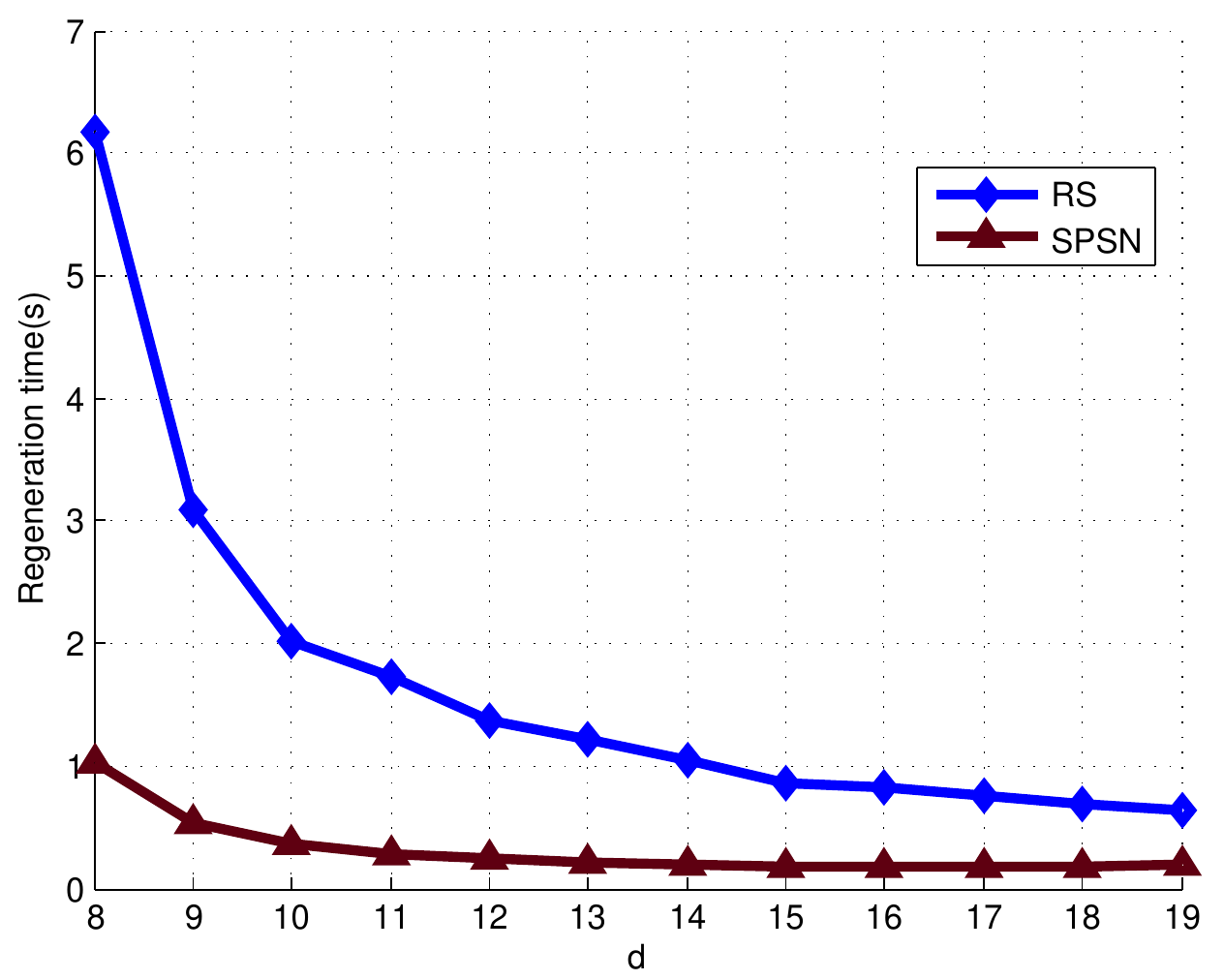}
   \label{Fig:3}}
  \subfigure[Different number of storage nodes holding coded data blocks: other parameters are N=1000, k=8, d=10 ]{\includegraphics[width=0.333\textwidth]{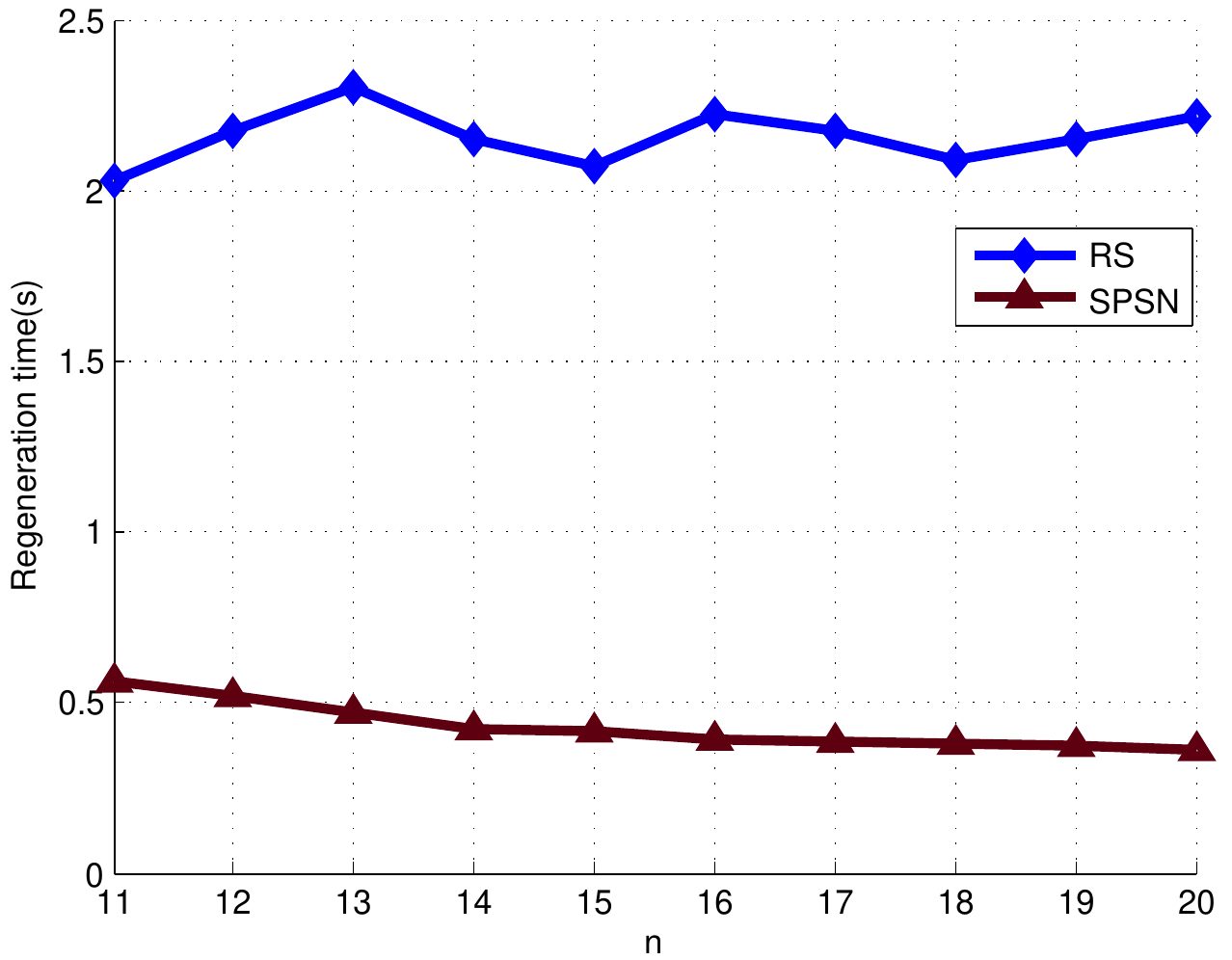}
   \label{Fig:4} }
\caption{Regeneration time of topologies resulted from RS and SPSN for different system parameters. Link capacity varies between [10,120]Mbps in this part of simulation. M=100Mb.}
\end{figure*}


We need to select $d$ providers from the $n-1$ survival nodes for data regeneration. To test how the number of storage nodes in $V_p$ affects the regeneration time, we fix $d=10$ as previously and change the value of $n$. Fig.~\ref{Fig:4} shows the simulation results. At each value of $n$, SPSN reduces the regeneration time dramatically than RS. Providers and newcomer are obtained in an absolutely random way by RS. Regeneration time of the scheme RS fluctuates as the value of $n$ changes. Because the number of providers $d$ and bandwidth range $[10,120]$Mbps is fixed, it fluctuates around $2.17s$. 
For SPSN, there will be more choice of the $d$ providers as $n$ turns larger. Its regeneration time declines obviously as $n$ increases in the figure.

\subsection{Impact of flexible end-to-end traffic}

In this part of evaluation, we test the impact of flexible end-to-end traffic on regeneration time for the same settings.

\subsubsection{Benefit comparison between the two approaches}

To compare the benefit brought from the flexible amount of data blocks generated by each provider and the node selection schemes, we test FRS (RS with flexible end-to-end traffic), and our node selection scheme SPSN for the same original file object. FLEX algorithm, selecting the providers and the newcomer with flexible $\beta^*$, is also considered here. From Fig.~\ref{Fig:7} we can see that SPSN enjoys shorter regeneration time than FRS at every link capacity distribution. At the most contrasting link capacity range $[0.3,120]$Mbps, the simple flexible method needs $3.53 \times$ regeneration time compared with SPSN. By the results we may conclude that nodes selection approach can reduce the regeneration time more radically. The scheme FLEX performs a bit better than SPSN. Flexible amount of data blocks computed according to the available link capacity further cuts down the regeneration time in heterogeneous environment. When bandwidth heterogeneity becomes modest from $[50,120]$Mbps, the gap between the two schemes narrows, even can be overlooked at $[90,120]$Mbps.

\begin{figure}[!htb]
\centering
\hspace*{-0.15in}\includegraphics[width = 8cm]{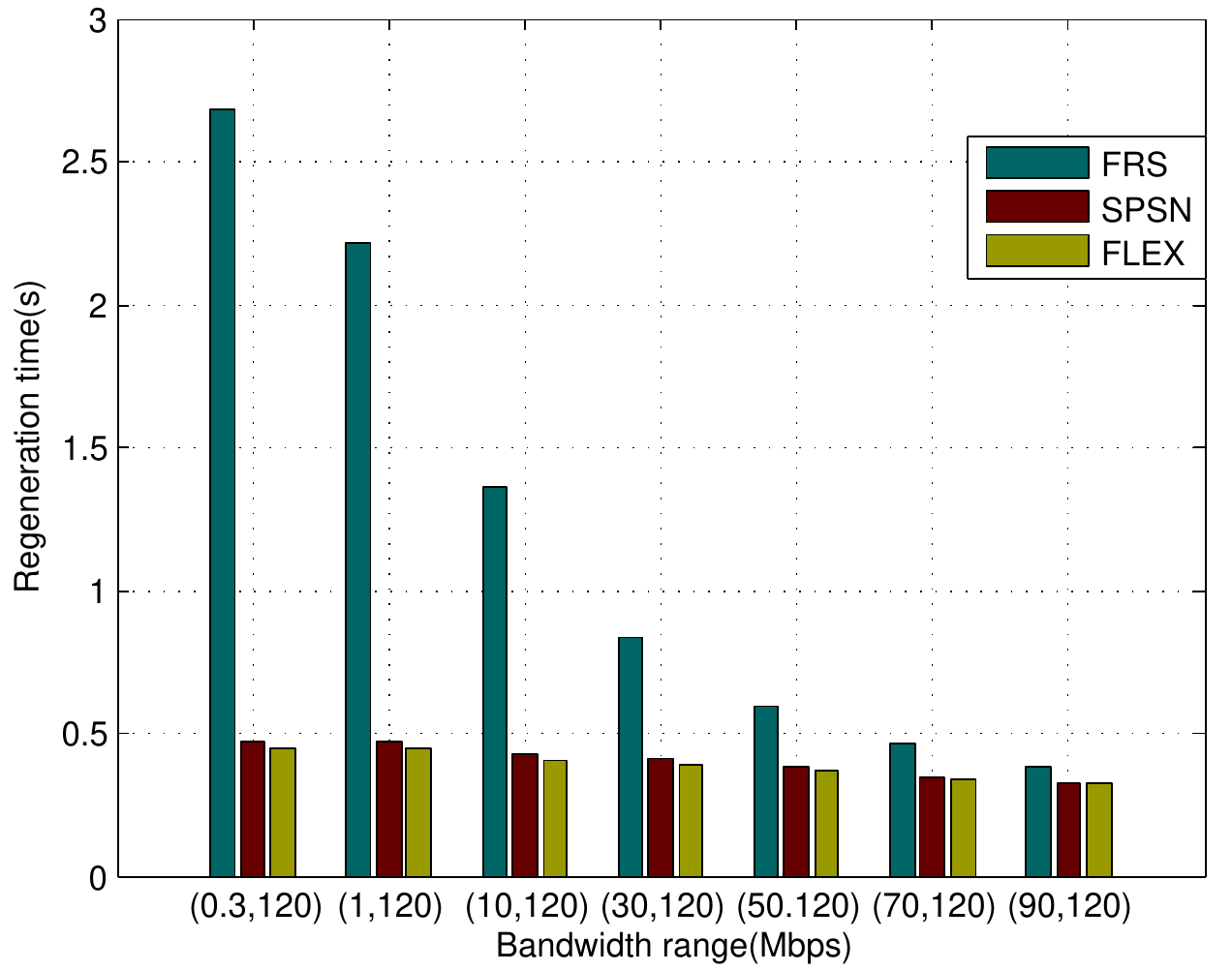}
\caption{Regeneration time of topologies resulted from FRS, SPSN and FLEX for different link capacity distributions. The parameters are N=1000, n=14, k=8, d=10, and M=100Mb.}
\label{Fig:7}
\end{figure}

\subsubsection{Evaluation of FLEX}

\begin{figure}[!htb]
\centering
\hspace*{-0.15in}\includegraphics[width = 9cm]{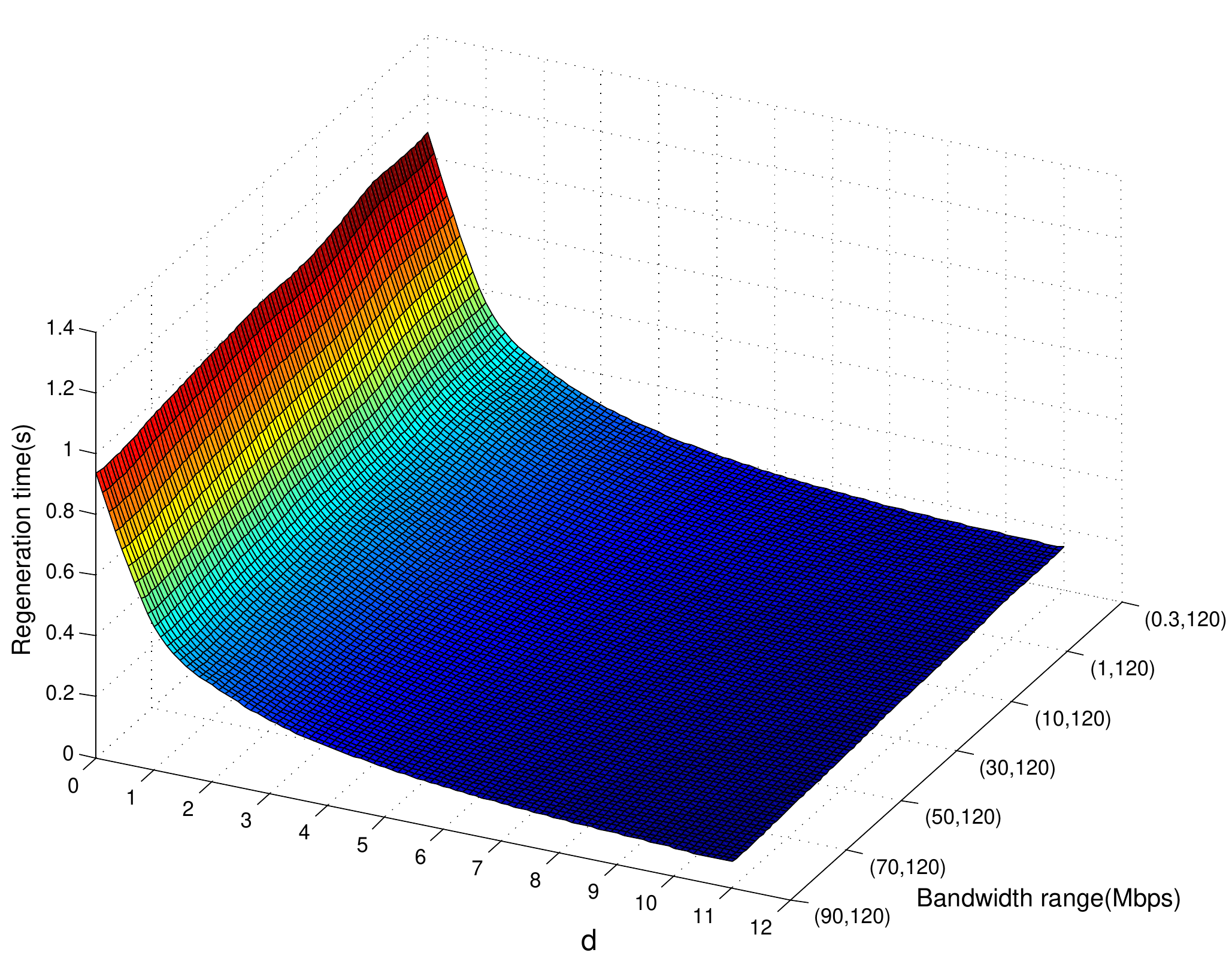}
\caption{Regeneration time of topologies resulted from FLEX for different link capacity distributions and the number of providers. The parameters are N=1000, n=14, k=8, and M=100Mb.}
\label{Fig:8}
\end{figure}

We further study the performance of FLEX for bandwidth heterogeneity and the value of $d$. Results are shown in Fig.~\ref{Fig:8}. It can be found that FLEX presents a similar trend as the scheme SPSN. Regeneration time of FLEX reduces as the number of providers increases, and as the bandwidth heterogeneity weakens.

\subsection{Evaluation of SPSN-F}

\begin{figure*}[!htbp]
   \subfigure[Different number of storage nodes in the system: coding parameters are n=14, k=8, d=10. ]{\includegraphics[width=0.33\textwidth]{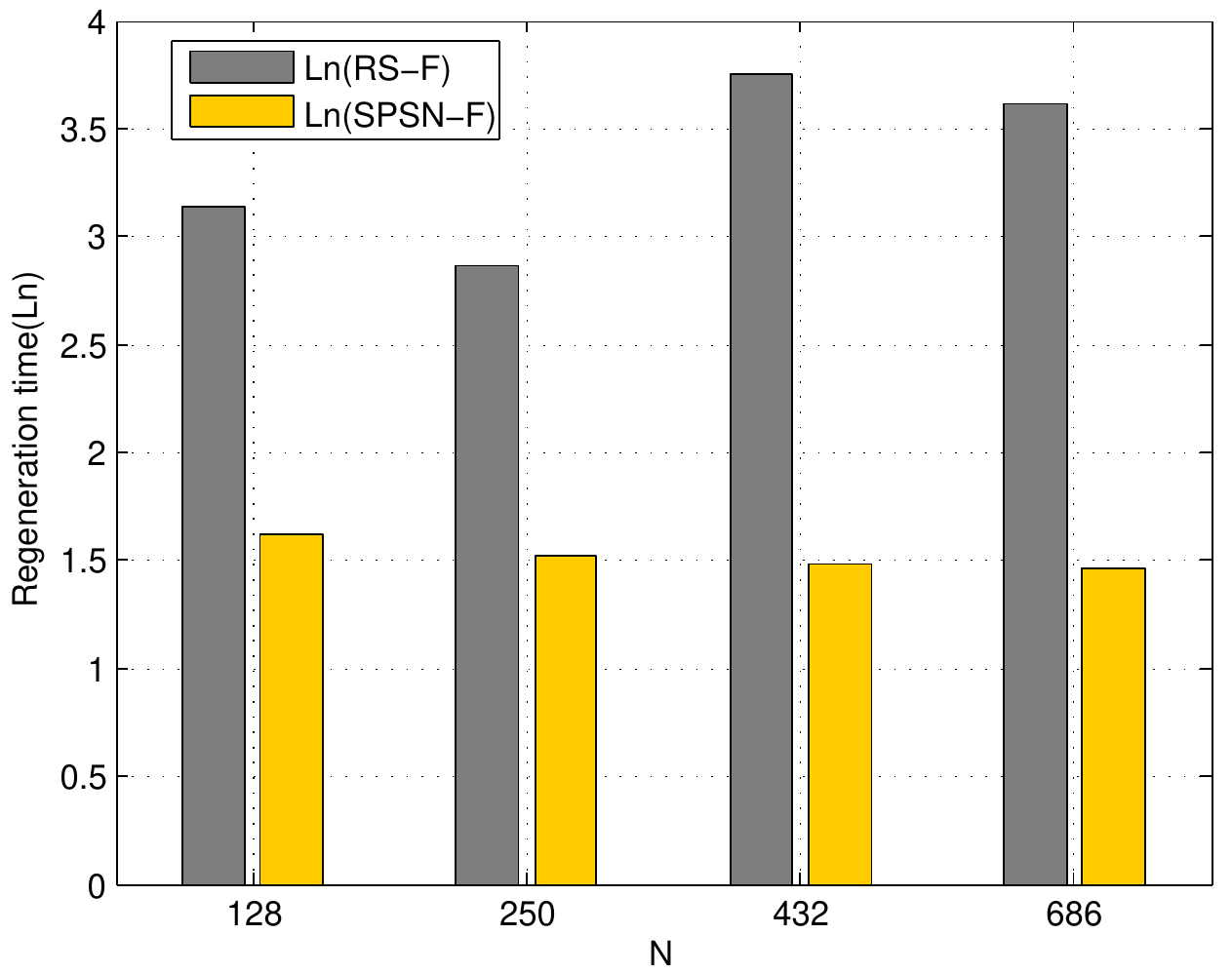}
   \label{Fig:9} }
  \subfigure[Different number of providers: other parameters are N=128, n=20, k=8, ]{\includegraphics[width=0.33\textwidth]{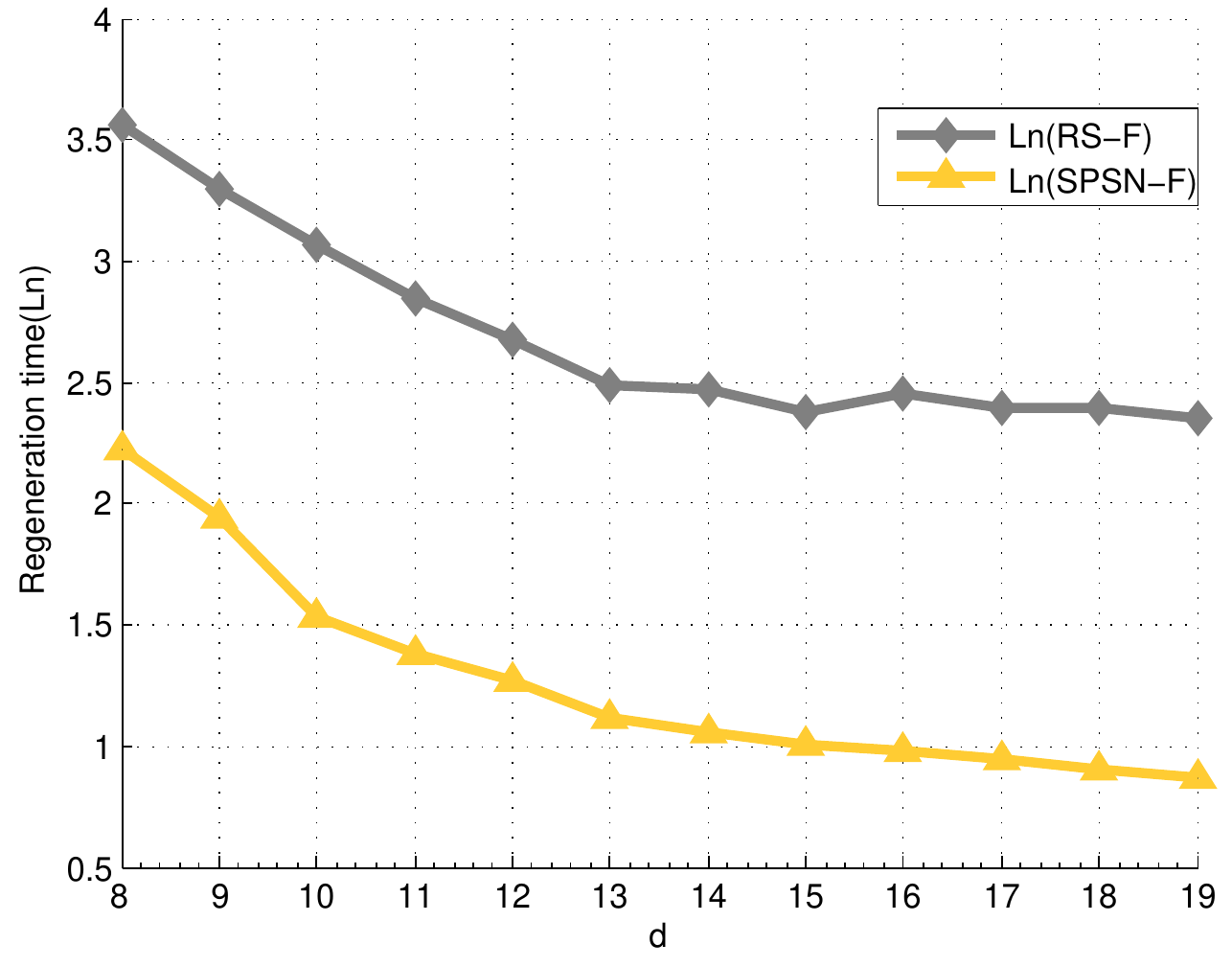}
   \label{Fig:10}}
  \subfigure[Different number of storage nodes holding coded data blocks: other parameters are N=128, k=8, d=10 ]{\includegraphics[width=0.33\textwidth]{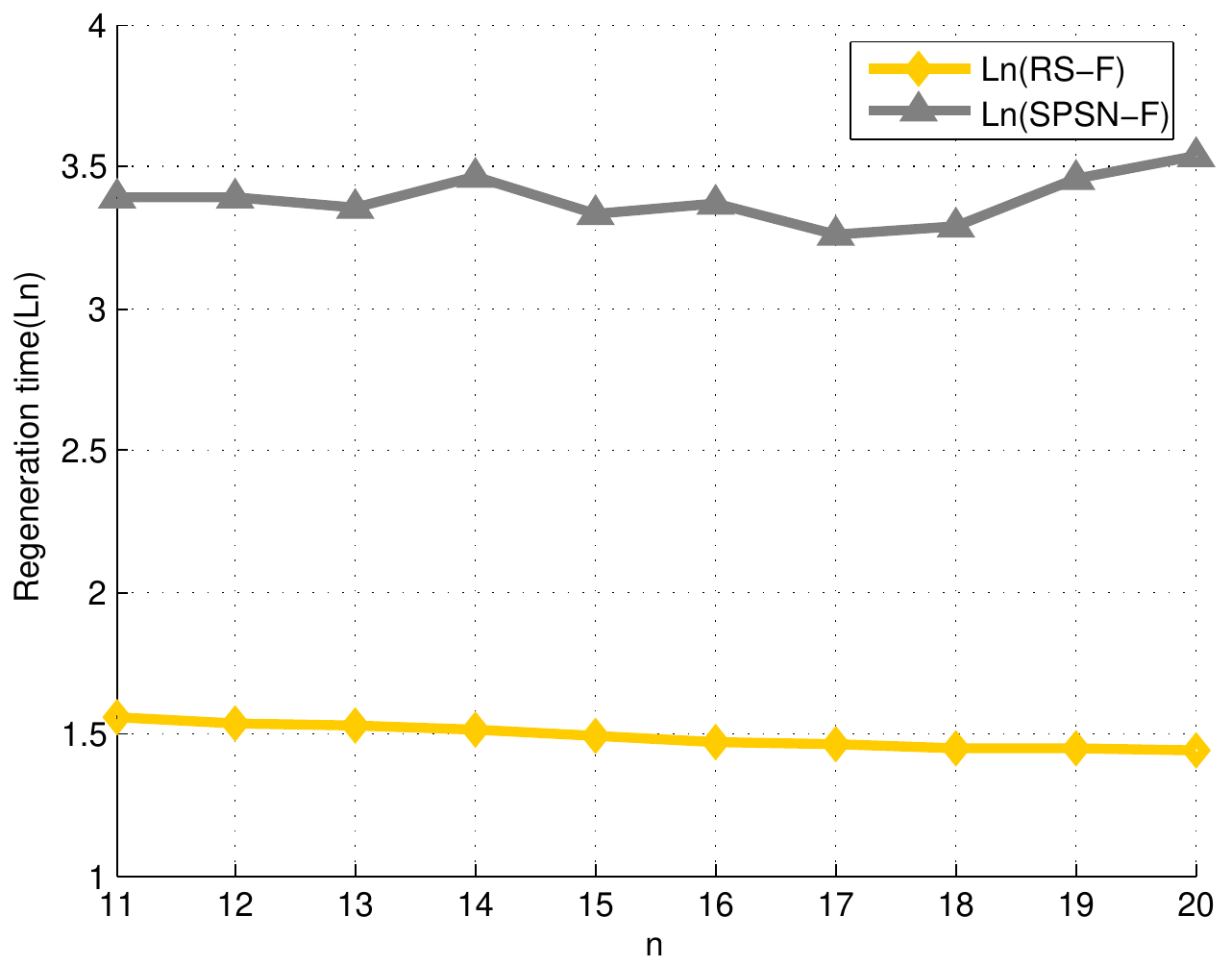}
   \label{Fig:11} }
\caption{Regeneration time of topologies resulted from RS-F and SPSN-F for different system parameters. M=100Mb.}
\end{figure*}

There are three tiers of links in Fat-tree architecture, shown in Fig.~\ref{Fig:5}. Physical bandwidths of links in each level differ orders of magnitude intuitively, moving up the network hierarchy \cite{benson2010network, benson2010understanding}. The bandwidths of links at the bottom layer, {\em i.e.,} links between edge switches and storage servers, are set to obey the uniform distribution on the interval $[1,120]$Mbps. And the bandwidths of links on the middle layer and the upper layer, are set to be five and ten times of the bottom layer.

Simulation results are shown in Fig.~\ref{Fig:9}-\ref{Fig:11}. We take the natural logarithm of the regeneration time to show the relative improvement. The regeneration time reduces obviously when SPSN-F is employed. We examine the impact of network scale, and the coding patterns to the regeneration time.

Network scale is measured by $K$ in Fat-tree architecture. There are $K^3/4$ storage nodes in the system. In Fig.~\ref{Fig:9}, the value of $K$ is 8, 10, 12, 14. Regeneration time decreases apparently at each network scale. Random selection scheme still presents fluctuating regeneration time because the network scale matters nothing to random operations. In the following experiment, $K$ is set to be 8, with 128 storage nodes in the Fat-tree topology.

In Fig.~\ref{Fig:10}, both of the two selection schemes show a decline trend as the parameter $d$ turns larger. This lies in the inherent property of regenerating codes. The amount of data transmitted from each provider $\beta$ will turn smaller when $d$ grows larger. The advantage brought about by SPSN-F can be recognized by the big gap between the two curves.

The number of storage nodes holding coded blocks affects $V_p$. We can select $d$ providers from a larger candidate set as $n$ grows larger. The regeneration time of SPSN-F exhibits a slightly downward trend in Fig.~\ref{Fig:11}. We can see that the selection of providers paly a less crucial role.

\section{Conclusion and Future Work} \label{sec: Conclusion}
In this paper, we focus on the problem of reducing regeneration time in overlay networks and real-world Fat-tree topologies. We analyze the data recovery process and find that the selection of participating nodes can seriously affect the regeneration time, especially in data center networks with heterogeneous link capacities. To reduce the regeneration time, we consider the selection of the newcomer and providers in different topologies. We have also incorporated the observation of flexible end-to-end traffic to enhance our node selection schemes and finally propose SPSN and FLEX algorithms to construct repairing topologies in overlay networks. Flows of data regeneration in real-world data center networks show special features. We analyze the regeneration time in Fat-tree architecture and propose the node selection algorithm SPSN-F to reduce the regeneration time. Experimental results show that the data regeneration time can be reduced obviously using our scheme, even greater in more heterogeneous networks, compared with random determination of the providers and newcomer. 

Optimal node selection schemes are proposed based on the available bandwidths between storage nodes in distributed storage systems. 
Experiments in this paper are based on NS2 as in the study of data center traffic conducted by T. Benson {\em et al.} \cite{benson2010understanding}. 
We aim to implement a prototype of our node selection scheme, and evaluate it in a real data center.

\normalsize



\begin{thebibliography}{10}
\providecommand{\url}[1]{#1}
\csname url@samestyle\endcsname
\providecommand{\newblock}{\relax}
\providecommand{\bibinfo}[2]{#2}
\providecommand{\BIBentrySTDinterwordspacing}{\spaceskip=0pt\relax}
\providecommand{\BIBentryALTinterwordstretchfactor}{4}
\providecommand{\BIBentryALTinterwordspacing}{\spaceskip=\fontdimen2\font plus
\BIBentryALTinterwordstretchfactor\fontdimen3\font minus
  \fontdimen4\font\relax}
\providecommand{\BIBforeignlanguage}[2]{{%
\expandafter\ifx\csname l@#1\endcsname\relax
\typeout{** WARNING: IEEEtran.bst: No hyphenation pattern has been}%
\typeout{** loaded for the language `#1'. Using the pattern for}%
\typeout{** the default language instead.}%
\else
\language=\csname l@#1\endcsname
\fi
#2}}
\providecommand{\BIBdecl}{\relax}
\BIBdecl

\bibitem{ghemawat2003google}
S.~Ghemawat, H.~Gobioff, and {\em et al.}, ``The google file system,'' in
  \emph{ACM SIGOPS Operating Systems Review}, vol.~37, no.~5, 2003, pp. 29--43.

\bibitem{shvachko2010hadoop}
K.~Shvachko, H.~Kuang, and {\em et al.}, ``The hadoop distributed file
  system,'' in \emph{Proceedings of MSST}, 2010, pp. 1--10.

\bibitem{a2013solution}
K.~V. Rashmi, N.~B. Shah, and {\em et al.}, ``A solution to the network
  challenges of data recovery in erasure-coded distributed storage systems: A
  study on the facebook warehouse cluster,'' \emph{preprint arXiv:1309.0186},
  2013.

\bibitem{weatherspoon2002erasure}
H.~Weatherspoon, J.~D. Kubiatowicz, and {\em et al.}, ``Erasure coding vs.
  replication: A quantitative comparison,'' in \emph{Peer-to-Peer
  Systems}.\hskip 1em plus 0.5em minus 0.4em\relax Springer, 2002, pp.
  328--337.

\bibitem{rhea2003pond}
S.~C. Rhea, P.~R. Eaton, and {\em et al.}, ``Pond: The oceanstore prototype,''
  in \emph{Proceedings of FAST}, vol.~3, 2003, pp. 1--14.

\bibitem{bhagwan2004total}
R.~Bhagwan, K.~Tati, and {\em et al.}, ``Total recall: System support for
  automated availability management,'' in \emph{Proceedings of NSDI}, vol.~4,
  2004, pp. 25--25.

\bibitem{huang2012erasure}
C.~Huang, H.~Simitci, and {\em et al.}, ``Erasure coding in windows azure
  storage,'' in \emph{Proceedings of USENIX Annual Technical Conference}, 2012,
  pp. 15--26.

\bibitem{dimakis2010network}
A.~G. Dimakis, P.~Godfrey, and {\em et al.}, ``Network coding for distributed
  storage systems,'' \emph{IEEE Transactions on Information Theory}, vol.~56,
  no.~9, pp. 4539--4551, 2010.

\bibitem{fat-tree2009}
A.~Greenberg, J.~R. Hamilton, and {\em et al.}, ``Vl2: a scalable and flexible
  data center network,'' in \emph{ACM SIGCOMM Computer Communication Review},
  vol.~39, no.~4, 2009, pp. 51--62.

\bibitem{fatreee2008al}
M.~Al-Fares, A.~Loukissas, and A.~Vahdat, ``A scalable, commodity data center
  network architecture,'' in \emph{ACM SIGCOMM Computer Communication Review},
  vol.~38, no.~4, 2008, pp. 63--74.

\bibitem{cloudcost2008cost}
A.~Greenberg, J.~Hamilton, and {\em et al.}, ``The cost of a cloud: research
  problems in data center networks,'' \emph{ACM SIGCOMM Computer Communication
  Review}, vol.~39, no.~1, pp. 68--73, 2008.

\bibitem{benson2010network}
T.~Benson, A.~Akella, and {\em et al.}, ``Network traffic characteristics of
  data centers in the wild,'' in \emph{Proceedings of SIGCOMM conference on
  Internet Measurement}, 2010, pp. 267--280.

\bibitem{benson2010understanding}
T.~Benson, A.~Anand, and {\em et al.}, ``Understanding data center traffic
  characteristics,'' \emph{ACM SIGCOMM Computer Communication Review}, vol.~40,
  no.~1, pp. 92--99, 2010.

\bibitem{googleArchitecture2003web}
L.~A. Barroso, J.~Dean, and {\em et al.}, ``Web search for a planet: The google
  cluster architecture,'' \emph{Micro, IEEE}, vol.~23, no.~2, pp. 22--28, 2003.

\bibitem{lee2005measuring}
S.~J. Lee, P.~Sharma, and {\em et al.}, ``Measuring bandwidth between planetlab
  nodes,'' in \emph{Passive and Active Network Measurement}.\hskip 1em plus
  0.5em minus 0.4em\relax Springer, 2005, pp. 292--305.

\bibitem{qygong2015}
Q.~Gong, J.~Wang, and {\em et al.}, ``Optimal node selection for data
  regeneration in heterogeneous distributed storage systems,'' in \emph{in
  Proceedings of the ICPP}, 2015.

\bibitem{cisco2007oversubscription}
\BIBentryALTinterwordspacing
Cisco data center infrastructure 2.5 design guide. [Online]. Available:
  \url{http://www.cisco.com/application/pdf/en/us/guest/netsol/ns107/c649/ccmigration_09186a008073377d.pdf}
\BIBentrySTDinterwordspacing

\bibitem{li2010tree}
J.~Li, S.~Yang, and {\em et al.}, ``Tree-structured data regeneration in
  distributed storage systems with regenerating codes,'' in \emph{Proceedings
  of INFOCOM}, 2010, pp. 1--9.

\bibitem{yan2014infocom}
Y.~Wang, D.~Wei, and {\em et al.}, ``Heterogeneity-aware data regeneration in
  distributed storage systems,'' in \emph{Proceedings of INFOCOM}, 2014, pp.
  1878--1886.

\bibitem{gaston2013realistic}
B.~Gast{\'o}n, J.~Pujol, and {\em et al.}, ``A realistic distributed storage
  system: the rack model,'' \emph{preprint arXiv:1302.5657}, 2013.

\bibitem{dimakis2011survey}
A.~G. Dimakis, K.~Ramchandran, and {\em et al.}, ``A survey on network codes
  for distributed storage,'' \emph{in Proceedings of the IEEE}, vol.~99, no.~3,
  pp. 476--489, 2011.

\bibitem{exactrepairMBR2009exact}
K.~V. Rashmi, N.~B. Shah, and {\em et al.}, ``Exact regenerating codes for
  distributed storage,'' Tech. Rep., 2009.

\bibitem{exactMSR2010exact}
C.~Suh, K.~Ramchandran, and {\em et al.}, ``Exact-repair mds codes for
  distributed storage using interference alignment,'' in \emph{Information
  Theory Proceedings (ISIT), 2010}, 2010, pp. 161--165.

\bibitem{exactMSR22010distributed}
V.~R. Cadambe, S.~A. Jafar, and {\em et al.}, ``Distributed data storage with
  minimum storage regenerating codes-exact and functional repair are
  asymptotically equally efficient,'' \emph{preprint arXiv:1004.4299}, 2010.

\bibitem{msrcode2015MSR}
K.~V. Rashmi, P.~Nakkiran, and {\em et al.}, ``Having your cake and eating it
  too: Jointly optimal erasure codes for i/o, storage and network-bandwidth,''
  in \emph{in Proceedings of FAST}, 2015.

\bibitem{meas1jsac2003evaluation}
N.~Hu and P.~Steenkiste, ``Evaluation and characterization of available
  bandwidth probing techniques,'' \emph{IEEE Journal on Selected Areas in
  Communications}, vol.~21, no.~6, pp. 879--894, 2003.

\bibitem{meas2003new}
C.~L.~T. Man, G.~Hasegawa, and M.~Murata, ``A new available bandwidth
  measurement technique for service overlay networks,'' in \emph{Management of
  Multimedia Networks and Services}.\hskip 1em plus 0.5em minus 0.4em\relax
  Springer, 2003, pp. 436--448.

\bibitem{computing}
\BIBentryALTinterwordspacing
{Clock Speed}. [Online]. Available:
  \url{http://ark.intel.com/products/83503/Intel-Core-i7-4980HQ-Processor-6M-Cache-up-to-4_00-GHz}
\BIBentrySTDinterwordspacing

\bibitem{equation2014survey}
\BIBentryALTinterwordspacing
Technical report: Elastic regenerating codes for geo-distributed cloud storage.
  [Online]. Available: \url{https://www.dropbox.com/s/xsenyan2sbiyd1i/main.pdf}
\BIBentrySTDinterwordspacing

\bibitem{fan2009HDFSRAIDdiskreduce}
B.~Fan, W.~Tantisiriroj, and {\em et al.}, ``Diskreduce: Raid for
  data-intensive scalable computing,'' in \emph{Proceedings of the Annual
  Workshop on Petascale Data Storage}, 2009, pp. 6--10.

\bibitem{sathiamoorthy2013xoring}
M.~Sathiamoorthy, M.~Asteris, and {\em et al.}, ``Xoring elephants: Novel
  erasure codes for big data,'' in \emph{in Proceedings of the VLDB Endowment},
  vol.~6, no.~5, 2013, pp. 325--336.

\bibitem{patil2011droptail}
G.~Patil, S.~McClean, and {\em et al.}, ``Drop tail and red queue management
  with small buffers: stability and hopf bifurcation,'' \emph{ICTACT Journal on
  Communication Technology}, vol.~2, no.~2, pp. 339--344, 2011.

\bibitem{calder2011windows}
B.~Calder, J.~Wang, and {\em et al.}, ``Windows azure storage: a highly
  available cloud storage service with strong consistency,'' in \emph{in
  Proceedings of SOSP}, 2011, pp. 143--157.

\end{thebibliography}
\end{document}